\documentclass[letterpaper,11pt]{article}
\usepackage[utf8]{inputenc}

\pdfoutput=1

\usepackage{amsmath, amsthm, amssymb, thm-restate}
\usepackage{algorithmicx}
\usepackage[table,xcdraw]{xcolor}
\usepackage[ruled,vlined,linesnumbered]{algorithm2e}
\usepackage{setspace}
\usepackage{mathtools}
\usepackage[numbers]{natbib}
\usepackage{comment} % enables comment environment
\usepackage{tcolorbox} 
\usepackage{xfrac}
\usepackage{hyperref}
\usepackage{multirow}
\usepackage{caption}
\usepackage{bm}
\usepackage{newfloat}
\usepackage{enumitem}
\usepackage{dblfloatfix} 
\usepackage{wrapfig}
\usepackage{enumitem}

\usepackage{fancyhdr}

\usepackage[margin=1in]{geometry}

\allowdisplaybreaks

\definecolor{mygreen}{RGB}{10,150,110}
\definecolor{myred}{RGB}{150,10,20}

\hypersetup{
     colorlinks=true,
     citecolor= mygreen,
     linkcolor= myred
}

\renewcommand{\epsilon}{\varepsilon}

\DeclareMathOperator{\E}{\ensuremath{\normalfont \textbf{E}}}

\newcommand{\hiddencomment}[1]{}

\newcommand{\mc}[1]{\ensuremath{\mathcal{#1}}}
\newcommand{\wt}[1]{\ensuremath{\widetilde{#1}}}

\newcommand{\cev}[1]{\reflectbox{\ensuremath{\vec{\reflectbox{\ensuremath{#1}}}}}}

\newcommand{\true}[0]{\ensuremath{\textsc{True}}}
\newcommand{\false}[0]{\ensuremath{\textsc{False}}}

\newcommand{\RGMM}[0]{\ensuremath{\textup{RGMM}}}
\newcommand{\ThSC}[0]{\ensuremath{\textsf{ThSC}}}
\newcommand{\ThSCfull}[0]{Threshold Set Cover}
\newcommand{\SC}[0]{\ensuremath{\textsf{SC}}}
\newcommand{\ST}[0]{\ensuremath{\textsf{ST}}}
\newcommand{\VO}[0]{\ensuremath{\textsf{VO}}}
\newcommand{\EO}[0]{\ensuremath{\textsf{EO}}}

\def\sU{\mathcal{U}}
\def\sF{\mathcal{F}}
\def\sW{\mathcal{W}}
\def\sol{\mathrm{SOL}}

\DeclareMathOperator{\poly}{poly}

\usepackage[noabbrev,nameinlink]{cleveref}
\crefname{lemma}{Lemma}{Lemmas}
\crefname{theorem}{Theorem}{Theorems}
\crefname{property}{Property}{Properties}
\crefname{claim}{Claim}{Claims}
\crefname{result}{Result}{Results}
\crefname{definition}{Definition}{Definitions}
\crefname{observation}{Observation}{Observations}
\crefname{proposition}{Proposition}{Propositions}
\crefname{assumption}{Assumption}{Assumptions}
\crefname{line}{Line}{Lines}
\crefname{figure}{Figure}{Figures}
\creflabelformat{property}{(#1)#2#3}
\crefname{equation}{}{}
\crefname{section}{Section}{Sections}
\crefname{appendix}{Appendix}{Appendices}
\crefname{algCounter}{Algorithm}{Algorithms}
\Crefname{algCounter}{Algorithm}{Algorithms}

\newtheorem{theorem}{Theorem}

\newtheorem{lemma}{Lemma}[section]
\newtheorem{proposition}[lemma]{Proposition}
\newtheorem{corollary}[lemma]{Corollary}

\newtheorem{definition}[lemma]{Definition}
\newtheorem{claim}[lemma]{Claim}

\newtheorem{observation}[lemma]{Observation}
\newtheorem{remark}[lemma]{Remark}
\newtheorem*{remark*}{Remark}

\definecolor{mylightgray}{RGB}{230,230,230}

\algnewcommand{\IIf}[2]{\textbf{if} #1 \textbf{then} #2}
\algnewcommand{\EndIIf}{\unskip\ \algorithmicend\ \algorithmicif}

\newenvironment{whitetbox}{
\par\addvspace{0.1cm}
\begin{tcolorbox}[width=\textwidth,
                  boxsep=5pt,
                  left=1pt,
                  right=1pt,
                  top=2pt,
                  bottom=2pt,
                  boxrule=1pt,
                  arc=0pt,
                  colframe=black,
                  colback=white
                  ]%%
}{
\end{tcolorbox}
}

\newcounter{algCounter}

\providecommand{\email}[1]{\href{mailto:#1}{\nolinkurl{#1}\xspace}}

\makeatletter
\renewcommand{\paragraph}{%
  \@startsection{paragraph}{4}%
  {\z@}{10pt}{-1em}%
  {\normalfont\normalsize\bfseries}%
}
\makeatother

\makeatletter
\patchcmd{\@algocf@start}% <cmd>
  {-1.5em}% <search>
  {0pt}% <replace>
  {}{}% <success><failure>
\makeatother

 \title{Sublinear Metric Steiner Tree via Improved Bounds for Set Cover}

\author{Sepideh Mahabadi\thanks{Microsoft Research. E-mail: \email{smahabadi@microsoft.com}.} \and Mohammad Roghani\thanks{Stanford University. E-mail: \email{roghani@stanford.edu}. The work was done while the author was an intern at Microsoft Research.} \and Jakub Tarnawski\thanks{Microsoft Research. E-mail: \email{jakub.tarnawski@microsoft.com}.} \and Ali Vakilian\thanks{Toyota Technological Institute at Chicago (TTIC). 
E-mail: \email{vakilian@ttic.edu}.}}

\date{}

\begin{document}

\maketitle

\thispagestyle{empty}

\begin{abstract}
    We study the metric Steiner tree problem in the sublinear query model. In this problem, for a set of $n$ points $V$ in a metric space given to us by means of query access to an $n\times n$ matrix $w$, and a set of terminals $T\subseteq V$, the goal is to find the minimum-weight subset of the edges that connects all the terminal vertices.

    \medskip
    Recently, Chen, Khanna and Tan [SODA'23] gave an algorithm that uses $\widetilde{O}(n^{13/7})$ queries and outputs a $(2-\eta)$-estimate of the metric Steiner tree weight, where $\eta>0$ is a universal constant. A key component in their algorithm is a sublinear algorithm for a particular set cover problem where, given a set system $\mathcal(\mathcal{U}, \mathcal{F})$, the goal is to provide a multiplicative-additive estimate for $|\mathcal{U}|-\textsf{SC}(\mathcal{U}, \mathcal{F})$. Here $\mathcal{U}$ is the set of elements, $\mathcal{F}$ is the collection of sets, and $\textsf{SC}(\mathcal{U}, \mathcal{F})$ denotes the optimal set cover size of $(\mathcal{U}, \mathcal{F})$. In particular, their algorithm returns a $(1/4, \varepsilon\cdot|\mathcal{U}|)$-multiplicative-additive estimate for this set cover problem using $\widetilde{O}(|\mathcal{F}|^{7/4})$ membership oracle queries (querying whether a set $S \in \mathcal{S}$ contains an element $e \in \mathcal{U}$), where $\varepsilon$ is a fixed constant. 

    \medskip
    In this work, we improve the query complexity of $(2-\eta)$-estimating the metric Steiner tree weight to $\widetilde{O}(n^{5/3})$ by showing a $(1/2, \varepsilon \cdot |\mathcal{U}|)$-estimate for the above set cover problem using $\widetilde{O}(|\mathcal{F}|^{5/3})$ membership queries. 
    To design our set cover algorithm, we estimate the size of a random greedy maximal matching for an auxiliary multigraph that the algorithm constructs implicitly, without access to its adjacency list or matrix. Previous analyses of random greedy maximal matching have focused on simple graphs, assuming access to their adjacency list or matrix. To address this, we extend the analysis of Behnezhad [FOCS'21] of random greedy maximal matching on simple graphs to multigraphs, and prove additional properties that may be of independent interest.
\end{abstract}

{
\clearpage
\hypersetup{hidelinks}
% \tableofcontents{}
\vspace{1cm}
\renewcommand{\baselinestretch}{0.1}
\setcounter{tocdepth}{2}
\thispagestyle{empty}
\clearpage
}

\setcounter{page}{1}
\section{Introduction}
In the Steiner tree problem, we are given an undirected graph $G = (V, E)$, where each edge $e$ has an associated cost $w(e)$, and a specified set of terminal vertices $T \subseteq V$. Then, the objective is to find a minimum-cost subgraph $H$ of $G$ that connects all terminals in $T$. The Steiner tree problem is one of the most fundamental problems in combinatorial optimization and has been extensively studied by the TCS community since it was included among Karp's 21 NP-Complete problems~\cite{Karp72}. The state-of-the-art approximation factor for the Steiner tree problem is $\ln{4} + \varepsilon <1.39$~\cite{byrka2010improved}, and it is known that approximating it to a factor better than $96/95$ is NP-hard~\cite{chlebik2008steiner}. This Steiner tree problem has been studied in various domains, including approximation algorithms~\cite{robins2005tighter,byrka2010improved}, online algorithms~\cite{imase1991dynamic,awerbuch2004line,megow2016power,gupta2014online}, stochastic algorithms~\cite{gupta2005stochastic,gupta2007stochastic,GargGLS08}, and massive data analysis models~\cite{chazelle2005approximating,czumaj2009estimating,chen2023query}.

\begin{definition}[Sublinear Metric Steiner Tree]\label{def:steiner-tree}
In the metric Steiner tree problem, we are given a set of points $V$, a set of terminal points $T\subseteq V$, and query access to an oracle $\mathcal{O}$ to the $|V|\times |V|$ distance matrix of a metric space $(V, w)$, where $\mathcal{O}(u,v)$ returns the weight $w(u,v)$ of the edge $(u,v)$.

Let $\ST(V, T, w)$ denote the weight of a minimum-weight Steiner tree on instance $(V, T, w)$.    
Then, the goal is to design an algorithm that estimates $\ST(V, T, w)$ using the fewest possible queries to the distance matrix via the oracle $\mathcal{O}$.
\end{definition}

Czumaj and Sohler~\cite{czumaj2009estimating} presented the first sublinear query algorithm for the metric Steiner tree problem, showing a $(2+\varepsilon)$-approximation using $\Tilde{O}(k/\varepsilon^{O(1)})$ queries through their improved algorithm for the minimum spanning tree (MST) problem. Specifically, this follows their sublinear $(1+\varepsilon)$-approximation for MST together with the well-known result by Gilbert and Pollak~\cite{gilbert1968steiner} showing that an $\alpha$-approximation for MST over the metric induced on the terminals $T$ is a $(2\alpha)$-approximation for the metric Steiner tree instance with $T$ as the terminal set. 

Recently, Chen, Khanna, and Tan~\cite{chen2023query} studied the design of sublinear algorithms with strictly better-than-$2$ approximation for the metric Steiner tree problem. On the lower bound side, they showed that for any $\varepsilon > 0$, estimating the Steiner tree cost to within a $(5/3 - \epsilon)$-factor requires $\Omega(n^2)$ queries, even when the number of terminals $|T|$ is constant. Moreover, they showed that for any $\varepsilon > 0$, estimating the Steiner tree cost to within a $(2 - \epsilon)$-factor requires $\Omega(n + |T|^{6/5})$ queries. Additionally, they proved that for any $0 < \varepsilon < 1/3$, any algorithm that outputs a $(2-\epsilon)$-approximate Steiner tree (not just its cost) requires $\Omega(n|T|)$ queries. On the upper bound side, they showed that it is possible to achieve a better-than-$2$ estimate of the Steiner tree cost in sublinear time: there exists an algorithm that, with high probability, computes a $(2-\eta)$-approximation of the Steiner tree cost using $\wt{O}(n^{13/7})$ queries, where $\eta > 0$ is a universal constant. At the core of their sublinear algorithm for metric Steiner tree with improved approximation guarantee, they relate the problem of achieving a better-than-$2$ estimation for the Steiner tree to a variant of set cover problem with a different objective.

\begin{definition}[\ThSCfull]\label{def:threshold-SetCover}
Given a universe of elements $\sU$ and a collection $\sF$ of subsets of $\sU$, in the \ThSCfull{} problem the goal is to estimate $\ThSC(\sU, \sF) \coloneqq |\sU| - \SC(\sU, \sF)$, where $\SC(\sU, \sF)$ denotes the size of an optimal set cover solution for $(\sU, \sF)$, i.e.,  the minimal number of sets in $\sF$ whose union equals $\sU$.
\end{definition}
Following the notation of~\cite{chen2023query} and for simplicity, in our technical sections, we also refer to this problem as set cover.

Specifically, given access to the adjacency matrix of the graph representation of $(\sU, \sF)$, where there is an edge between $e \in \sU$ and $S \in \sF$ if and only if $e \in S$, Chen, Khanna, and Tan~\cite{chen2023query} designed an algorithm that, for any constant $0 < \varepsilon < 1$, with high probability, outputs a {\em multiplicative-additive} $(1/4,\varepsilon |\sU|)$-approximation for estimation of $\ThSC(\sU, \sF)$ using $\wt{O}_{\varepsilon}(|\sF|^{3/2} + |\sF|^{3/4} \cdot |\sU|)$ queries to the adjacency matrix (or, membership queries). An estimate $\sol$ for \ThSCfull{} on $(\sU, \sF)$ is a {\em multiplicative-additive} $(\gamma_1, \gamma_2)$-approximation, if $\gamma_1 \cdot \ThSC(\sU, \sF) - \gamma_2 \le \sol \le \ThSC(\sU, \sF)$.   
 
More broadly, there has been a large body of work on solving set cover problems in the massive data models of computation over the past decade \cite{saha2009maximum, demaine2014streaming, Har-PeledIMV16, emek2016semi, assadi2016tight, indyk2017fractional, assadi2017tight, bateni2017almost, IndykMRVY18, GrunauMRV20}. In particular the work of \cite{IndykMRVY18, GrunauMRV20} consider the set cover problem in the sublinear query model. However their algorithms assumes that it has an access to the adjacency list model as opposed to the adjacency matrix model, and thus cannot be directly employed here.

\subsection{Our Results}
Our key contribution is an algorithm for \ThSCfull{}, offering improved approximation guarantees and query complexity, as detailed below:
\begin{restatable}[Our Algorithm for \ThSCfull{}]{theorem}{maintheorem}\label{thm:setcover-general}
There exists an algorithm that,
given a set system $(\mc{U}, \mc{F})$ with oracle access to its adjacency matrix (also known as membership queries), outputs a multiplicative-additive $(1/2, \epsilon \cdot |\mc{U}|)$-approximation to \ThSCfull{},
%maximizing $\ThSC(\mc{U}, \mc{F}) = |\mc{U}| - \SC(\mc{U}, \mc{F})$,
in $\wt{O}(|\mc{F}|^{5/3})$ time, with high probability. 
\end{restatable}

Note that both the query complexity and the running time of the algorithm are bounded by $\wt{O}(|\mc{F}|^{5/3})$, improving upon the algorithm by Chen, Khanna, and Tan~\cite{chen2023query} for large values of $|\sU|$, which uses $\wt{O}_{\varepsilon}(|\sF|^{3/2} + |\sF|^{3/4} \cdot |\sU|)$ membership queries and provides a multiplicative-additive $(1/4, \varepsilon |\mc{U}|)$-approximation for the problem. Notably, when $|\mc{U}| = \omega(|\mc{F}|^{2/3})$, the algorithm becomes sublinear in $|\mc{U}| \cdot |\mc{F}|$, making it especially relevant for applications in the metric Steiner tree problem. More specifically, our new algorithm for \ThSCfull{} results in the following improved sublinear query algorithm for the metric Steiner tree problem, which we show in \cref{sec:steiner-tree}.

\begin{restatable}[Sublinear Algorithm for Metric Steiner Tree]{theorem}{maintheoremsteiner}\label{thm:steiner-tree}
There exists an algorithm that,
given an instance of metric Steiner tree denoted by $(V,T,w)$
with oracle access $\mathcal{O}$ to the distance matrix of $(V,w)$,
outputs a $(2-\eta)$-estimate of $\ST(V,T,w)$ using $\wt{O}(n^{5/3})$ queries to $\mathcal{O}$,
where $\eta > 0$ is a universal constant,
with high probability.
\end{restatable}

Notably, the query complexity of our algorithm improves upon the $\wt{O}(n^{13/7})$ query complexity of the algorithm of~Chen, Khanna, and Tan~\cite{chen2023query}.

For a detailed overview of our technical contribution, see~\Cref{sec:tech-overview}.

\section{Technical Overview}\label{sec:tech-overview}

In this section, we provide a brief overview of the technical challenges involved in designing our algorithms. To design a sublinear time algorithm for the Steiner tree problem, we use the framework developed by \citet*{chen2023query}. They demonstrated that breaking the 2-approximation barrier for the Steiner tree problem can be reduced to solving an instance of a set cover problem. We refer the reader to Section 4.1 of \cite{chen2023query} for details on this reduction.

We denote the variant of the set cover problem as {\em \ThSCfull{}}. Given a collection of sets $\mc{F}$ over a universe of elements $\mc{U}$, we aim to estimate the value of $\ThSC(\sU, \sF) = |\mc{U}| - \SC(\mc{U}, \mc{F})$, where $\SC(\mc{U}, \mc{F})$ denotes the size of the optimal set cover of the given instance. To achieve our goal of breaking the 2-approximation barrier for the Steiner tree problem, we need to estimate $\ThSC(\sU, \sF)$ with a $(\gamma, \epsilon \cdot |\mc{U}|)$ multiplicative-additive error, where $\gamma$ must be a constant and $\epsilon$ is any (small enough) constant. For this problem, we only have access to a membership oracle of the instance, meaning we can query whether a particular element $e$ is in a particular set $S$ or not. Note that this type of access is generally considered more challenging compared to an adjacency list oracle, where the algorithm can access either the $i$th element of a set $S$, or the $i$th set containing an element $e$. The reason is that if an element is included in only a constant number of sets, the algorithm is required to spend $\Omega(|\mc{F}|)$ queries to find just one set that contains the element. Consequently, we cannot use the results from the literature on sublinear set cover \cite{GrunauMRV20, IndykMRVY18} because they all rely on an adjacency list access model.

We will now provide an informal, step-by-step description of our algorithm for \ThSCfull{}, highlighting its differences, innovations, and technical challenges in comparison to the algorithm of \citet*{chen2023query}. For simplicity, in this technical overview we assume that $|\mc{F}| = \wt{\Theta}(|\mc{U}|)$, since this represents the worst-case scenario for the Steiner tree problem. However, our formal proof does not depend on this assumption. We let $n = |\mc{F}|$.

\paragraph{First step: sparsification of the \ThSCfull{} instance.} The goal of this step is to produce a new instance where each element and each set has a low degree—specifically, where each element is in only a few sets, and each set contains only a few elements. This step is standard in designing sublinear algorithms for the set cover problem for different access models, and a slightly different version of it is also used in the algorithm by \cite{chen2023query}. Our slight modification of the sparsification step allows us to relax some constraints in the reduction from the Steiner tree problem to \ThSCfull, enabling us to achieve the same query complexity for both problems.

Let $x > 0$ be some constant that we optimize later. Consider a set $S \in \mc{F}$ and suppose we randomly sample $\wt{O}(n^{1-x})$ elements from the universe and query the membership of all the sampled elements in $S$. If the size of $S$ is at least $\wt{\Omega}(n^{x})$, we expect to see a large intersection. Conversely, if the size of $S$ is much smaller, we expect to see a small intersection. If a large intersection exists, we can remove the set $S$ and all its elements from the instance. Since this event occurs at most $\wt{O}(n^{1-x})$ times, we can account for the removed elements and sets using the additive error in our estimation. Similarly, we can show that all elements belonging to more than $\wt{\Omega}(n^{x})$ sets can be covered by a random subcollection of sets of size $\wt{O}(n^{1-x})$. Therefore, without loss of generality, by spending $\wt{O}(n^{2 - x})$ time, we can assume that each set contains at most $\wt{O}(n^{x})$ elements, and each element is included in at most $\wt{O}(n^{x})$ sets.

\paragraph{Second step: constructing an auxiliary graph $H$ and estimating the size of its maximum matching.} Similar to \cite{chen2023query}, we construct a graph $H$ with a vertex set where each vertex corresponds to an element of $\mc{U}$. We connect two vertices if their corresponding elements appear together in at least one set from $\mc{F}$. It is important to note that we do not construct $H$ explicitly, as doing so would be computationally expensive and require $\Omega(n^2)$ time. As shown by \cite{chen2023query}, if the size of the maximum matching of $H$ is large, it is evident that $\ThSC(\mc{U}, \mc{F})$ is significantly smaller than $|\mc{U}|$. Conversely, if the size of the maximum matching of $H$ is close to zero, then $\ThSC(\mc{U}, \mc{F})$ is also close to zero. This is sufficient for our purposes, as our goal is to obtain a constant-factor approximation. Intuitively, each matching edge in $H$ indicates that there are two elements that can be covered together, which increases the value of $\ThSC(\mc{U}, \mc{F})$.

There is extensive literature on estimating the size of maximum matching in sublinear time \cite{AzarmehrBR24, Behnezhad21, BehnezhadRR23a, BehnezhadRR23b, BehnezhadRR24, BhattacharyaKS23, BhattacharyaKSW23, KapralovMNT20, LeviRY15,  OnakSODA12, ParnasR07, YoshidaYISTOC09}, with significant progress made in recent years. For our application, we use the algorithm of \citet*{Behnezhad21} to estimate the size of a random greedy maximal matching (\RGMM) of the graph. In summary, this algorithm can estimate the size of the \RGMM{} of a graph in $\wt{O}(\bar{d})$ time \underline{if given access to the adjacency list} of the graph, where $\bar{d}$ denotes the average degree of the graph. We can now use this algorithm as a black box: 
\begin{itemize}
    \item The average degree of $H$ is $\wt{O}(n^{2x})$, since each element is in $\wt{O}(n^{x})$ sets and each set contains $\wt{O}(n^{x})$ elements.
    \item Each time the algorithm visits a vertex in $H$ (corresponding to an element), we can spend $\wt{O}(n^{1+x})$ time to find its adjacency list in $H$. This involves first querying all sets that include the element, and then making queries between those sets and all elements.
\end{itemize}
Therefore, we can simulate the algorithm from \cite{Behnezhad21} in $\widetilde{O}(\bar{d} \cdot n^{1+x}) = \widetilde{O}(n^{1+3x})$ time. By balancing this with the sparsification step, which requires $\widetilde{O}(n^{2-x})$ time, we can set $x = 1/4$ to achieve an algorithm with a running time of $\widetilde{O}(n^{7/4})$. This is essentially the running time of the algorithm by \citet*{chen2023query}.

\paragraph{Third step: using the algorithm of \citet{Behnezhad21} in a white-box manner.} To improve the running time of our algorithm, we need to open up the RGMM algorithm from \cite{Behnezhad21} and utilize its properties to apply it more effectively. The \RGMM{} algorithm is a local algorithm that explores the neighborhood of a given vertex to determine whether it is matched. A key observation is that during each exploration, the algorithm requires a random neighbor of the vertex that has not been explored yet. However, in the previous approach, we constructed the entire adjacency list of the vertex, which is redundant and inefficient. Intuitively, we only need to randomly identify one of the vertex's neighbors in each step.

However, the first challenge we encounter is that we cannot select a neighbor uniformly at random. To illustrate this, consider the following example. 
Suppose that we have five elements $\mc{U} = \{e_1, \ldots, e_5\}$ and three sets: $S_1 = \{e_1, e_2, e_3 \}$, $S_2 = \{ e_1, e_2, e_4\}$, and $S_3 = \{e_1, e_2, e_5\}$. Suppose that we want to find a random neighbor of $e_1$ in $H$. If we first find all sets that include $e_1$ and then query between those sets and all elements uniformly at random until we find an edge in $H$, we are likely to see the edge $(e_1,e_2)$ because it appears in all sets. Consequently, the algorithm has a bias towards finding neighbors that appear in more sets with the element.

To overcome this challenge, rather than defining an auxiliary simple graph $H$, we define an auxiliary multigraph $H$. In this multigraph, if two elements appear in the same set multiple times, we add an edge for each of those occurrences. Note that the average degree of $H$ remains at most $\wt{O}(n^{2x})$. However, the algorithm and analysis for \RGMM{} from \cite{Behnezhad21} are designed for simple graphs. We extend these results to multigraphs and show that we can estimate the size of \RGMM{} for a multigraph, given access to its adjacency list. This extension may be of independent interest and could be useful for tackling other problems in sublinear time. To establish this, we build on the exquisite approach first introduced by \citet*{YoshidaYISTOC09} and further explored in various settings \cite{Behnezhad21, BehnezhadRRS-SODA23, TSP-icalp24}. We employ techniques such as the analysis of the round-complexity of maximal independent sets \cite{FischerN18}, double-counting arguments to bound the average complexity of RGMM on multigraphs, and others; we encourage the reader to refer to \Cref{sec:rgmm} for further details. 

Now, suppose that for each vertex the \RGMM{} algorithm explores in \( H \), we first query all sets to identify those that include the corresponding element. Since the algorithm explores at most \(\widetilde{O}(\bar{d})\) vertices in \( H \), this step will cost at most \(\widetilde{O}(\bar{d} \cdot n) = \widetilde{O}(n^{1+2x})\) in total. Let $v$ be a vertex that the \RGMM{} algorithm is exploring at the moment. Define $\mc{S}_v$ to be the collection of sets that include element $v$. Now, if we query uniformly at random between all elements and the collection $\mc{S}_v$, each incident edge of $v$ in the multigraph $H$ has an equal probability of being sampled, which resolves the first challenge. For now, assume that the degree of all vertices in \( H \) is \(\bar{d}\). Since \(|S_v| = \widetilde{O}(n^{x})\) and there are \(\widetilde{O}(n)\) elements in total, we expect to find an element in one of the sets of \(S_v\) every \(\widetilde{O}(n^{1+x}/\bar{d})\) queries. Thus, to identify a random neighbor of a vertex in \(H\), we need to spend \(\widetilde{O}(n^{1+x}/\bar{d})\) time. The \RGMM{} algorithm queries for a random neighbor of a vertex \(\widetilde{O}(\bar{d})\) times, since the exploration size is \(\widetilde{O}(\bar{d})\); therefore, the total cost is \(\widetilde{O}(n^{1+x})\). Combining this with the cost of sparsification, the total cost of the algorithm is \(\widetilde{O}(\max(n^{2-x}, n^{1+2x}))\), which is \(\widetilde{O}(n^{5/3})\) if we set \(x = 1/3\).

The second challenge arises because the \RGMM{} algorithm may predominantly visit vertices with a very low degree in \(H\). For such vertices, finding a random neighbor can be much more time-consuming. Generally, if a vertex \(v\) in \(H\) has degree \(\deg_H(v)\), then each time the algorithm finds a random neighbor of \(v\), it needs to spend \(\widetilde{O}(n^{1+x}/\deg_H(v))\) time. Therefore, if the algorithm frequently encounters vertices with constant degree, each query to find a random neighbor can cost \(\widetilde{O}(n^{1+x})\). With the exploration size being \(\widetilde{O}(\bar{d})\), this can significantly increase the query complexity of the algorithm. As a property of the local \RGMM{} algorithm, we demonstrate that each vertex is visited by the algorithm in proportion to its degree in \(H\). More formally, we prove that \RGMM{} requires \(\widetilde{O}(\deg_H(v)/n)\) neighbors of \(v\) on average, for a uniformly random permutation of edges. Thus, the degree-dependent factors cancel each other out, and the average cost of this part can be upper-bounded by \(\widetilde{O}(n^{1+x})\), which is enough for us to get the \(\widetilde{O}(n^{5/3})\) running time for the \ThSCfull.

\section{Preliminaries}\label{sec:preliminaries}

As is common in the literature, we use the term ``with high probability'' to refer to a probability of at least $1-n^{-\alpha}$, for a sufficiently large constant $\alpha \geq 2$. Moreover, we use $\wt{O}(\cdot)$, $\wt{\Theta}(\cdot)$, and $\wt{\Omega}(\cdot)$ to hide the dependency on $\poly(\log n)$. For a maximization problem of estimating some value $\chi$, and for $\gamma_1 \in (0, 1]$ and $\gamma_2 > 0$, we say that $\wt\chi$ is a multiplicative-additive $(\gamma_1, \gamma_2)$-approximation of the value $\chi$ if $\gamma_1\chi - \gamma_2 \leq \wt\chi \leq \chi$.

\subsection{Probabilistic Tools}
We use the following standard concentration inequalities in our proof.

\begin{proposition}[Chernoff Bound]
    Let $X_1, X_2, \ldots, X_n$ be $n$ independent Bernoulli random variables. Let $X = \sum_{i=1}^n X_i$. For any $k > 0$, it holds that
    \begin{align*}
        \Pr[|X - \E[X]| \geq k] \leq 2 \exp \left(- \frac{k^2}{3\E[X]}\right).
    \end{align*}
\end{proposition}

\begin{definition}[Negative Association \cite{kumarDevProschen, saxenaKhursheed, wajc2017negative}]
    Let $X_1, X_2, \ldots, X_n$ be a set of random variables. We say this set is negatively associated if for any two disjoint index sets $I, J \subseteq [n]$, and two functions $f$ and $g$, both either monotonically increasing or monotonically decreasing, the following condition is satisfied:
    \begin{align*}
        \E[f(X_i: i \in I) \cdot g(X_j: j \in J)] \leq \E[f(X_i: i \in I)] \cdot \E[g(X_j: j \in J)].
    \end{align*}
\end{definition}

\begin{proposition}[Chernoff Bound for Negatively Associated Variables]\label{pro:NA-chernoff}
    Let $X_1, X_2, \ldots, X_n$ be a set of negatively associated Bernoulli random variables. Let $X = \sum_{i=1}^n X_i$. Then
    \begin{align*}
        \Pr\left[X \geq (1+\alpha) \E[X]\right] \leq \left( \frac{e^\alpha}{(1+\alpha)^{1+\alpha}} \right)^{\E[X]}
        \end{align*}
and
    \begin{align*}        
        \Pr\left[X \leq (1-\alpha) \E[X]\right] \leq \left( \frac{e^{-\alpha}}{(1-\alpha)^{1-\alpha}} \right)^{\E[X]}
    \end{align*}
\end{proposition}

\begin{proposition}[Markov Inequality]
    Let $X$ be a non-negative random variable. For any $\alpha > 0$, it holds that
    \begin{align*}
        \Pr[X \geq \alpha] \leq \frac{\E[X]}{\alpha}.
    \end{align*}
\end{proposition}

\subsection{Graph Theory}

A {\em multigraph} is a type of graph in which multiple edges, also known as parallel edges, are allowed between any pair of vertices. A {\em line graph} of a graph $G$ is a graph that represents the adjacencies between the edges of $G$. More formally, given a graph $G$, the line graph $L(G)$ is constructed as follows:
\begin{itemize}
    \item \textbf{Vertices}: Each vertex in $L(G)$ corresponds to an edge in $G$.
    \item \textbf{Edges}: Two vertices in $L(G)$ are connected by an edge if and only if their corresponding edges in $G$ share a common endpoint (i.e., they are incident to the same vertex in $G$).
\end{itemize}
For a graph $G$, we let $\deg_G(v)$ be the degree of vertex $v$. In the case of multigraphs, $\deg_G(v)$ counts parallel edges multiple times.

\paragraph{Random Greedy Maximal Matching (RGMM):} Given a graph \(G = (V, E)\), a random greedy maximal matching is constructed by first selecting a random permutation \(\pi\) of the edges \(E\). The algorithm then iterates over the edges in the order specified by \(\pi\), adding each edge to the matching if neither of its endpoints is already matched (i.e., if none of its adjacent edges have been included in the matching so far). This process continues until all edges have been considered. The term "random" comes from the fact that the permutation \(\pi\) is chosen uniformly at random among all possible permutations of the edges.

\paragraph{Random Greedy Maximal Independent Set (MIS):} Given a graph \(G = (V, E)\), a random greedy maximal independent set is constructed by first selecting a random permutation \(\pi\) of the vertices \(V\). The algorithm then iterates over the vertices in the order specified by \(\pi\), adding each vertex to the independent set if none of its neighbors are already in the set (i.e., it does not share an edge with any vertex already included in the independent set). This process continues until all vertices have been considered. The term "random" comes from the fact that the permutation \(\pi\) is chosen uniformly at random among all possible permutations of the vertices.

\paragraph{Parallel Randomized Greedy Maximal Independent Set:} Let \(G\) be a graph, and let \(\pi\) be a permutation of its vertices. In each iteration, we select all vertices whose rank is lower than that of all their neighbors and then remove these vertices along with their neighbors from the graph. The number of iterations required for \(G\) to become empty is referred to as the round complexity, denoted by \(\rho(G, \pi)\). It is not hard to see that the MIS produced by the parallel randomized greedy maximal independent set is the same as the output of the random greedy maximal independent set for a fixed permutation $\pi$.

\section{Sublinear Algorithm for Set Cover}\label{sec:set-cover}

In this section, we formalize and analyze our algorithm for the set cover variant described above. Throughout this section, we assume that $\mc{U}$ denotes the universe and $\mc{F}$ denotes the collection of sets. We will slightly abuse notation by letting $k = |\mc{U}|$ and $n = |\mc{F}|$. Without loss of generality, we can assume that $n \geq k$.\footnote{For the sake of this problem, for each element in the universe we can add a set that only contains the element. The same assumption is also made in \cite{chen2023query}.} We use $\SC(\mc{U}, \mc{F})$ for the size of the minimum set cover of the input instance. Our goal is to design an algorithm that estimates the value of $\chi = k - \SC(\mc{U}, \mc{F})$ with at most $\epsilon k$ additive error and a constant multiplicative factor. For $\gamma_1 \in (0, 1]$ and $\gamma_2 > 0$, we say that $\wt\chi$ is a multiplicative-additive $(\gamma_1, \gamma_2)$-approximation of the value $\chi$ if $\gamma_1\chi - \gamma_2 \leq \wt\chi \leq \chi$. Similar to the reduction from the Steiner tree problem to set cover in \cite{chen2023query}, we need to estimate $k - \SC(\mc{U}, \mc{F}_{\neq 2})$ where $\mc{F}_{\neq 2}$ denotes the collection of all sets in $\mc{F}$ except those of size exactly 2. For simplicity, we focus on estimating $\chi$, and in the final step of this section, we will explain how to handle sets of size 2 in our algorithm. For our application, we need $\gamma_1$ to be constant and $\gamma_2 = \epsilon k$ where $\epsilon$ is a small fixed constant.

\paragraph{A High-Level Description of the Algorithm:} Our algorithm for estimating the value of $\chi$ is formalized in \Cref{alg:set-cover}. Apart from the collection of sets $\mc{F}$ and the universe of elements $\mc{U}$, the algorithm runs with two parameters, $\alpha$ and $\beta$, both of which can be determined based on the values of $x$ and $y$, which we optimize in the final step of the analysis. The algorithm has three phases: 1) sparsification of the sets, 2) sparsification of the elements, and 3) estimating the size of a maximum matching of the auxiliary graph $H$ (\Cref{def:graph-h}). In the following, we first describe each component in words before moving on to the formal proofs.

\paragraph{(Step 1) Set Sparsification:} This component of the algorithm is formalized in \Cref{alg:sparsification-sets}. The algorithm maintains a collection of sets $\hat{\mc{F}}$ and a universe of elements $\hat{\mc{U}}$ that are initially equal to $\mc{F}$ and $\mc{U}$, respectively. We iterate over all sets in $\mc{F}$ one by one and for each set $S$, we sample $r_1 = |\hat{\mc{U}}|/\alpha$ random elements from $\hat{\mc{U}}$. Intuitively, if $|S \cap \hat{\mc{U}}| \ge \wt\Omega(\alpha)$, we expect to have an element of $S$ in the $r_1$ random sampled  elements. Having this in mind, if there is a large enough intersection ($\Omega(\log n)$) between the sampled elements and $S$, we remove set $S$ and all its elements from $\hat{\mc{F}}$ and $\hat{\mc{U}}$, respectively (we add this set to our solution). Therefore, after the execution of the algorithm, each remaining set in $\hat{\mc{F}}$ has at most $\wt{O}(\alpha)$ elements. On the other hand, if $|S \cap \hat{\mc{U}}|$ is smaller than $\alpha$, we expect to see a small intersection with the $r_1$ sampled elements. Consequently, the number of times the algorithm removes a set and its elements from $\hat{\mc{F}}$ and $\hat{\mc{U}}$ is at most $k / \alpha = o(k)$, which can be accounted for by the additive error in the estimation. Also, if at any point the size of the maintained universe becomes smaller than some threshold (in the algorithm the value of the threshold is $\wt{\Theta}(\alpha)$, the algorithm stops processing the rest of the sets (\Cref{ln:end-condition}) since $\hat{\mc{U}} = \wt{O}(\alpha)$. This step differs from the algorithm in \cite{chen2023query} because we sequentially sparsify the sets, whereas their approach is non-adaptive.

\paragraph{(Step 2) Sparsification of Elements:} This part of the algorithm is formalized in \Cref{alg:sparsification-elements}. Similar to the previous step, we want to sparsify our instance such that each element in the remaining instance appears in at most $\wt{O}(\beta)$ sets. Let $\hat{\mc{U}}$ and $\hat{\mc{F}}$ be the output of \Cref{alg:sparsification-sets}. We sample $r_2 = |\hat{\mc{U}}|/\beta$ random sets from the collection $\hat{\mc{F}}$. With the same intuition as the previous step, if some element is in at least $\wt\Omega(\beta)$ sets of $\hat{\mc{F}}$, we expect to see it in many sampled sets. We partition the elements of $\hat{\mc{U}}$ into $\mc{U}_{low}$ and $\mc{U}_{high}$, depending on whether their intersection with the randomly sampled elements is smaller than a given threshold or not. With high probability, each element in $\mc{U}_{low}$ appears in at most $\wt{O}(\beta)$ sets of $\hat{\mc{F}}$, and each element in $\mc{U}_{high}$ appears in at least $\wt\Omega(\beta)$ sets of $\hat{\mc{F}}$. This suffices to show that any random subset of $\hat{\mc{F}}$ of size $\epsilon k / 2$ can cover all elements of $\mc{U}_{high}$, which can be included in the additive error of the estimation. This step is similar to the approach used in \cite{chen2023query}.

After steps 1 and 2 of the algorithm, we have the property that each set in the remaining instance has at most $\wt{O}(\alpha)$ elements, and each element in the remaining instance is in at most $\wt{O}(\beta)$ sets.

\paragraph{(Step 3) Estimating the Maximum Matching of Auxiliary Graph $H$:} We construct an auxiliary multigraph $H$ with vertex set $\mc{U}_{low}$. For each set $S \in \hat{\mc{F}}$, we add an edge in $H$ between every two elements of $S$. Note that we do not explicitly construct the multigraph $H$ because doing so would require $\Omega(nk)$ time, which is not feasible. We now estimate the size of the maximum matching in $H$ to produce our final estimate. Intuitively, if $\chi$ is very small (close to zero), meaning that nearly $k$ sets are required to cover the universe, then the maximum matching in $H$ will also be small. To see this, note that a matching edge implies that its two endpoints can be covered by the same set. Conversely, if $\chi$ is large (almost equal to $k$), then $H$ will have a large matching because each set in the set cover solution covers many new elements, which can be almost paired up in $H$. Thus, obtaining a constant approximation for the size of the maximum matching in $H$ is sufficient to achieve our goal, and we do this by estimating the size of a random greedy maximal matching in $H$. To that end, we modify and analyze the algorithm from \cite{Behnezhad21} for multigraphs and adapt it to our access model for the multigraph $H$, as discussed in detail in \Cref{sec:rgmm}. The algorithm in \cite{chen2023query} constructs a similar graph, but it is not a multigraph in the sense that, for any two elements that appear together in multiple sets, only a single edge is added between them in $H$.

\begin{definition}[Auxiliary Multigraph $H$]\label{def:graph-h}
    Let $\hat{\mc{F}}$ and $\mc{U}_{low}$ be as defined in \Cref{ln:line4,ln:line3}, respectively, of \Cref{alg:set-cover}. We construct an auxiliary graph $H$ with vertex set $\mc{U}_{low}$ such that for each set $S$ in $\hat{\mc{F}}$ and each two different elements $e, e' \in \mc{U}_{low}$, we add an edge $(e, e')$ to $H$. Note that $H$ is a multigraph: multiple sets may contain both elements $e$ and $e'$, and we add an edge for each of these sets.
\end{definition}

\begin{algorithm}[H]
\caption{Sublinear Time Algorithm for Set Cover}
\label{alg:set-cover}
\textbf{Input:} Collection of sets $\mc{F}$ and universe of elements $\mc{U}$.

\textbf{Parameter:} $\alpha \gets n^x, \quad \beta \gets 10 \max(k/n^{1-y}, 1) \cdot n \log(n) / k$. \label{ln:def_of_beta}

Let $\hat{\mc{F}}$ and $\hat{\mc{U}}$ be the output of \Cref{alg:sparsification-sets} with input $\mc{F}$, $\mc{U}$ and $\alpha$. \label{ln:line3}

Let $\mc{U}_{low}$ and $\mc{U}_{high}$ be the output of \Cref{alg:sparsification-elements} with input $\hat{\mc{F}}$, $\hat{\mc{U}}$, and $\beta$. \label{ln:line4}

Let $H$ be the auxiliary multigraph defined in \Cref{def:graph-h}. \Comment{We do not build $H$ explicitly.}

Let $\wt\mu$ be the estimate of $\E_\pi | \RGMM(H, \pi) |$ using the algorithm in \Cref{sec:rgmm}.

Let $\wt\chi \gets \wt\mu + |\mc{U} \setminus \mc{U}_{low} | -\epsilon k / 2$. \label{ln:def_of_wtchi}

\Return $\wt\chi$.

\end{algorithm}

\begin{algorithm}[H]
\caption{Sparsification of Sets}
\label{alg:sparsification-sets}
\textbf{Input:} Collection of sets $\mc{F}$, universe of elements $\mc{U}$, and parameter $\alpha$ that controls the sparsification ratio.

$\hat{\mc{F}} \gets \mc{F}, \quad \hat{\mc{U}} \gets \mc{U}, \quad c \gets 0$. \Comment{We use $c$ only for the analysis.}

\For{$S \in \mc{F}$}{
    \If{$|\hat{\mc{U}}| < 10 \alpha \log n$ \label{ln:end-condition}}{
    \textbf{break}
    }

    $r_1 \gets |\hat{\mc{U}}| / \alpha$. \label{ln:sample-size-alg2} 

    Let $e_1, e_2, \ldots, e_{r_1}$ be $r_1$ random elements of $\hat{\mc{U}}$.

    Make queries between $S$ and elements $e_1, \ldots, e_{r_1}$.

    \If{$| \{e_1, \ldots , e_{r_1} \} \cap S| \geq 10\log n $ \label{ln:dense}}{
        $\hat{\mc{F}} \gets \hat{\mc{F}} \setminus \{S\}$.

        $c \gets c + 1$.

        \For{$e \in \hat{\mc{U}}$}{
            \If{$e \in S$}{
                $\hat{\mc{U}} \gets \hat{\mc{U}} \setminus \{e\}$.
            }
        }
    }
}
\Return $\hat{\mc{F}}, \hat{\mc{U}}$

\end{algorithm}

\begin{algorithm}[H]
\caption{Sparsification of Elements}
\label{alg:sparsification-elements}
\textbf{Input:} Collection of sets $\hat{\mc{F}}$, universe of elements $\hat{\mc{U}}$, and parameter $\beta$ that controls the sparsification ratio.

$r_2 \gets |\hat{\mc{U}}|/\beta$.

\If{$r_2 < 20 \log n / \epsilon$}{
    $\mc{U}_{low} \gets \hat{\mc{U}}, \quad \mc{U}_{high} \gets \emptyset$.
    
    \Return $\mc{U}_{low}$, $\mc{U}_{high}$.
}

Let $\{ S_1, S_2, \ldots, S_{r_2} \}$ be $r_2$ random sets from $\hat{\mc{F}}$.

Make queries between all elements in $\hat{\mc{U}}$ and sets in $\{ S_1, S_2, \ldots, S_{r_2} \}$.

Let $\mc{U}_{low}$ be the elements that appeared in at most $20 \log n / \epsilon$ many sets.

$\mc{U}_{high} \gets \hat{\mc{U}} \setminus \mc{U}_{low}$.

\Return $\mc{U}_{low}$, $\mc{U}_{high}$.

\end{algorithm}

\subsection{Proof of Correctness}

In this section, we prove the correctness of \Cref{alg:set-cover}.

\begin{claim}\label{clm:deletion-bound}
    For any set $S$ such that the condition of \Cref{ln:dense} holds during the execution of \Cref{alg:sparsification-sets}, with high probability it holds that $\lvert S \cap \hat{\mc{U}} \rvert \geq \alpha$, where $\hat{\mc{U}}$ denotes the universe that \Cref{alg:sparsification-sets} maintains at the time it processes $S$.
\end{claim}
\begin{proof}
Suppose that $S$ is a set such that   $\lvert S \cap \hat{\mc{U}} \rvert < \alpha$  at the time that \Cref{alg:sparsification-sets} processes this set. Let $X_i$ be the random variable that indicates $e_i \in S$. Thus, we have $\Pr[X_i] \leq |S \cap \hat{\mc{U}} | / |\hat{\mc{U}}|$. Let $X = \sum_{i = 1}^{r_1} X_i$. By linearity of expectation, we have $\E[X] < r_1 \alpha /|\hat{\mc{U}}| = 1$. Also, note that $X_i$'s are negatively associated random variables. Let $\lambda = (9 \log n) / \E[X]$. Therefore, using the Chernoff bound for negatively associated random variables (\Cref{pro:NA-chernoff}) we have
\begin{align*}
    \Pr[X \geq (1 + \lambda) \E[X]] &\leq \left( \frac{e^\lambda}{(1+\lambda)^{1+\lambda}}\right)^{\E[X]}\\
    & \leq \left( \frac{e^\lambda}{\lambda^{\lambda}}\right)^{\E[X]} & (\text{Since } \lambda > 1) \\
    & = \left( \frac{e}{\lambda}\right)^{9 \log n} & (\text{Since } \lambda = (9 \log n) / \E[X])\\
    & \leq \frac{1}{n^{9}} & (\text{Since } \lambda > e^2)
\end{align*}
which implies that with probability of at least $1 - n^{-9}$,
\begin{align*}
    X < (1+\lambda)\E[X] = \E[X] + 9\log n < 10\log n.
\end{align*}
Since we have $n$ sets, using a union bound, with a probability at least $1-n^{-8}$, for any set such that the condition of \Cref{ln:dense} holds, we have $\lvert S \cap \hat{\mc{U}} \rvert \geq \alpha$.
\end{proof}

\begin{claim}\label{clm:bound-on-if-cond}
Let $c$ be the variable used in \Cref{alg:sparsification-sets}. With high probability, we have $c \leq k / \alpha$.
\end{claim}

\begin{proof}
    By \Cref{clm:deletion-bound}, for every set $S$ such that the condition on \Cref{ln:dense} of \Cref{alg:sparsification-sets} holds, with high probability we have that $\lvert S \cap \hat{\mc{U}} \rvert \geq \alpha$. Hence, each time the algorithm increases $c$, the size of $\hat{\mc{U}}$ decreases by $\alpha$. Therefore, the total number of times the algorithm increases $c$ is upper-bounded by $k / \alpha$.
\end{proof}

\begin{claim}\label{clm:sparse-set-alg2}
    Let $\hat{\mc{F}}$ and $\hat{\mc{U}}$ be the output of \Cref{alg:sparsification-sets}. Then, each set $S \in \hat{\mc{F}}$ has at most $20 \alpha \log n$ elements in $\hat{\mc{U}}$ with high probability.
\end{claim}

\begin{proof}
    First, note that if the algorithm stops because of the condition of \Cref{ln:end-condition} and does not process $S$, it holds that $|\hat{\mc{U}}| < 10 \alpha \log n$ and the claim trivially holds.

    Let $S$ be a set such that at the time that \Cref{alg:sparsification-sets} processes $S$, we have $|S \cap \hat{\mc{U}}| \geq 20 \alpha \log n$. Similar to the proof of \Cref{clm:deletion-bound}, let $X_i$ be the random variable that indicates $e_i \in S$ and $X = \sum_{i=1}^{r_1} X_i$. Hence, $\E[X] \geq 20 r_1 \alpha \log n / \lvert \hat{\mc{U}} \rvert = 20 \log n$. Since $X_i$'s are negatively associated random variables, using Chernoff bound (\Cref{pro:NA-chernoff}) for $\lambda = (9 \log n) / \E[X]$, we have
    \begin{align*}
    \Pr[X \leq (1 - \lambda) \E[X]] &\leq \left( \frac{e^\lambda}{(1+\lambda)^{1+\lambda}}\right)^{\E[X]}\\
    & \leq \left( \frac{e^\lambda}{\lambda^{\lambda}}\right)^{\E[X]} & (\text{Since } \lambda > 1) \\
    & = \left( \frac{e}{\lambda}\right)^{9 \log n} & (\text{Since } \lambda = (9 \log n) / \E[X])\\
    & \leq \frac{1}{n^{9}} & (\text{Since } \lambda > e^2)
\end{align*}
which implies that with a probability of at least $1 - n^{-9}$,
\begin{align*}
    X > (1-\lambda)\E[X] = \E[X] - 9\log n \geq 20\log n - 9 \log n > 10\log n.
\end{align*}
Therefore, the condition on \Cref{ln:dense} of \Cref{alg:sparsification-sets} must hold for all such sets with a probability of at least $1 - n^{-8}$ using union bound.
\end{proof}

\begin{lemma}[Sets Sparsification Guarantee]\label{lem:set-sparsification-gaur}
    Let $\hat{\mc{F}}$ and $\mc{U}_{low}$ be outputs of \Cref{alg:sparsification-sets} and \Cref{alg:sparsification-elements}. Also, let $S \in \hat{\mc{F}}$. Then, with high probability, $S$ contains at most $\widetilde{O}(\alpha)$ elements of $\mc{U}_{low}$.
\end{lemma}
\begin{proof}
    Note that $\mc{U}_{low} \subseteq \hat{\mc{U}}$ where $\hat{\mc{U}}$ is the output of \Cref{alg:sparsification-sets}. Also, by \Cref{clm:sparse-set-alg2}, we have $| S \cap \hat{\mc{U}}| \leq 20 \alpha \log n$. Thus, with high probability we have $| S \cap \mc{U}_{low}| \leq 20 \alpha \log n = \widetilde{O}(\alpha)$.
\end{proof}

\begin{lemma}[Elements Sparsification Guarantee]\label{lem:element-sparsification}
    Let $\mc{U}_{low}$ be as output by \Cref{alg:sparsification-elements} and let $e \in \mc{U}_{low}$. Then, with high probability, there are at most $\widetilde{O}(\beta)$ sets of $\hat{\mc{F}}$ that contain $e$.
\end{lemma}
\begin{proof}
    Let $e$ be an element of $\hat{\mc{U}}$, where $\hat{\mc{U}}$ is the output of \Cref{alg:sparsification-sets}. We show that if at least $40 \beta \log n / \epsilon$ sets in $\hat{\mc{F}}$ contain $e$, then $e \in \mc{U}_{high}$ in the output of \Cref{alg:sparsification-elements} with high probability.

    Let $X_i$ be an indicator variable for $S_i$ containing $e$. Thus, we have $\E[X_i] \geq 40 \beta \log n / (\epsilon n)$. Define $X = \sum_{i=1}^{r_2} X_i$. Hence, $\E[X] \geq r_2 \cdot 40 \beta \log n / (\epsilon n) = 40 \log n / \epsilon$. Since $X_i$'s are negatively associated random variables, using Chernoff bound (\Cref{pro:NA-chernoff}) for $\lambda = (9 \log n) / \E[X]$, we have
    \begin{align*}
    \Pr[X \leq (1 - \lambda) \E[X]] &\leq \frac{1}{n^{9}} & (\text{Similar to the proof of \Cref{clm:sparse-set-alg2}})
\end{align*}
which implies that $X > 20 \log n$ with a probability of at most $1 - n^{-9}$, since
\begin{align*}
    X > (1-\lambda)\E[X] = \E[X] - 9\log n \geq 40\log n / \epsilon - 9 \log n > 20\log n / \epsilon.
\end{align*}
Using a union bound for all elements in $\hat{\mc{U}}$ that are in at least $40\beta \log n / \epsilon$ sets of $\hat{\mc{F}}$, with high probability all of them are going to be included in $\mc{U}_{high}$. As a result, for each $e \in \mc{U}_{low}$, $e$ is in at most $\widetilde{O}(\beta)$ sets of $\hat{\mc{F}}$.
\end{proof}

\begin{claim}\label{clm:covering-high}
    Let $\mc{F}'$  be a random collection of $\epsilon k / 5$ sets of $\hat{\mc{F}}$. Then, with high probability, every element in $\mc{U}_{high}$ is in one of the sets of $\mc{F}'$.
\end{claim}
\begin{proof}
    Let $e \in \hat{\mc{U}}$ be such that at most $15 \beta \log n / \epsilon$ sets of $\hat{\mc{F}}$ contain $e$. Similar to the proof of \Cref{lem:element-sparsification}, let $X$ be a random variable that denotes the number of sets $S_i$ (for $i = 1,2,...,r_2$) that contain $e$. Using a Chernoff bound, we can prove that with a probability of at least $1-n^{-2}$, for all such elements $e$, we have $X < 20\log n / \epsilon$. Therefore, all elements in $\mc{U}_{high}$ are in at least $15 \beta \log n / \epsilon$ sets of $\hat{\mc{F}}$.

    Consequently, when $|\mc{F}'| \geq \epsilon k/5$, the expected number of sets in $\mc{F}'$ that cover element $e \in \mc{U}_{high}$ is at least
    \[
    \frac{15 \beta \log n}{\epsilon} \cdot \frac{1}{n} \cdot \frac{\epsilon k}{5} \ge \frac{15 n \log n}{\epsilon k} \cdot \frac{1}{n} \cdot \frac{\epsilon k}{5} = 3 \log n
    \]
    where we used that $\beta \ge n/k$ (see \Cref{ln:def_of_beta} in \Cref{alg:set-cover}). Using Chernoff bounds again, we expect all these elements to be covered by at least one set with high probability, which concludes the proof.
\end{proof}

\begin{claim}\label{clm:bound-on-mu}
    Let $\SC(\mc{U}_{low}, \hat{\mc{F}})$ be the optimal set cover size for the universe of elements $\mc{U}_{low}$ and the collection of sets $\hat{\mc{F}}$ which are the outputs of \Cref{alg:sparsification-elements} and \Cref{alg:sparsification-sets}, respectively. Let $\wt\mu$ be the output of \Cref{lem:rgmm} (i.e., a size estimate of a random greedy maximal matching of $H$) with $\epsilon k/2$ additive error. Then, it holds that 
    \begin{align*}
        \frac{1}{2}\left(|\mc{U}_{low}|- \SC(\mc{U}_{low}, \hat{\mc{F}}) \right) - \frac{\epsilon k}{2} \leq \wt\mu \leq |\mc{U}_{low}| - \SC(\mc{U}_{low}, \hat{\mc{F}}).
    \end{align*}
\end{claim}

\begin{proof}
Let $M$ be any maximal matching of $H$. First, we show that \[ \frac{1}{2}\left(|\mc{U}_{low}|- \SC(\mc{U}_{low}, \hat{\mc{F}})\right)  \leq |M| \leq |\mc{U}_{low}| - \SC(\mc{U}_{low}, \hat{\mc{F}}) . \] If for each edge in $M$, we take the corresponding set, and for the rest of the elements we take a set that only covers that element, we will have covered all elements of $\mc{U}_{low}$ with at most $|\mc{U}_{low}| - |M|$ sets, which implies that $|M| \leq |\mc{U}_{low}| - \SC(\mc{U}_{low}, \hat{\mc{F}})$. On the other hand, no set can cover two elements that are unmatched by $M$ at the same time. Thus, we have $|\mc{U}_{low}| - 2|M| \leq \SC(\mc{U}_{low}, \hat{\mc{F}})$.

Since the output of $\Cref{lem:rgmm}$ can be smaller than the minimum-sized maximal matching by an additive factor of at most $\epsilon k / 2$, we get the claim.
\end{proof}

\begin{claim}\label{clm:bound-low-cover-to-all}
    Conditioning on the high-probability events of \Cref{clm:bound-on-if-cond} and \Cref{clm:covering-high}, it holds that $\SC(\mc{U}, \mc{F}) \leq \SC(\mc{U}_{low}, \hat{\mc{F}}) + \epsilon k/2$.
\end{claim}
\begin{proof}
    All elements of $\mc{U} \setminus (\mc{U}_{low} \cup \mc{U}_{high})$ can be covered using the $c$ sets that are deleted in  \Cref{alg:sparsification-sets}. Moreover, by \Cref{clm:bound-on-if-cond}, we have $c < o(k)$ since $\alpha > \omega(1)$. On the other hand, the elements of $\mc{U}_{high}$ can be covered using $\epsilon k / 5$ sets, as shown in \Cref{clm:covering-high}. Finally, the elements of $\mc{U}_{low}$ can be covered using $\SC(\mc{U}_{low}, \hat{\mc{F}})$ sets. Therefore, the union of these three collections covers all elements, and the size of this solution is at most $\SC(\mc{U}_{low}, \hat{\mc{F}}) + \epsilon k / 2$, which completes the proof.
\end{proof}

\begin{lemma}\label{lem:approx-gaurantee}
    Conditioning on the high-probability events of \Cref{clm:bound-on-if-cond} and \Cref{clm:covering-high}, it holds that $\chi/2 - \epsilon k \leq \wt\chi \leq \chi$.
    (Recall that $\chi = k - \SC(\mc{U}, \mc{F})$ and $\wt\chi$ is the output of \Cref{alg:set-cover}, defined on its \Cref{ln:def_of_wtchi}.)
\end{lemma}
\begin{proof}
    We have that 
    \begin{align*}
        \wt \chi &= \wt \mu + | \mc{U} \setminus \mc{U}_{low}| - \frac{\epsilon k}{2}  & (\text{Output of \Cref{alg:set-cover}})\\
        & \leq |\mc{U}_{low}| - \SC(\mc{U}_{low}, \hat{\mc{F}}) + | \mc{U} \setminus \mc{U}_{low}| - \frac{\epsilon k}{2}  & (\text{By \Cref{clm:bound-on-mu}}) \\
        & = |\mc{U}| - \SC(\mc{U}_{low}, \hat{\mc{F}}) - \frac{\epsilon k}{2} \\
        & \leq |\mc{U}| - \SC(\mc{U}, \mc{F}) & (\text{By \Cref{clm:bound-low-cover-to-all}})\\
        & = \chi .
    \end{align*}
On the other hand,
\begin{align*}
    2(\wt \chi + \epsilon k)&= 2\left(\wt \mu + | \mc{U} \setminus \mc{U}_{low}|  + \frac{\epsilon k}{2} \right)\\
    & \geq  2\left( \frac{|\mc{U}_{low}|}{2} - \frac{\SC(\mc{U}_{low}, \hat{\mc{F}})}{2} +  | \mc{U} \setminus \mc{U}_{low}|\right) & (\text{By \Cref{clm:bound-on-mu}})\\
    & \geq |\mc{U}| - \SC(\mc{U}_{low}, \hat{\mc{F}})\\
    & = |\mc{U}| - \SC(\mc{U}_{low}, \mc{F})\\
    & \geq |\mc{U}| - \SC(\mc{U}, \mc{F})\\
    & = \chi ,
\end{align*}
which concludes the proof.
\end{proof}

\subsection{Time Complexity}

\begin{claim}\label{clm:alg2-time}
\Cref{alg:sparsification-sets} runs in $O(nk/\alpha)$ time with high probability.
\end{claim}
\begin{proof}
    First note that \Cref{alg:sparsification-sets} samples at most $k/\alpha$ elements for each set in $\mc{F}$ (\Cref{ln:sample-size-alg2}). So if we ignore the block of if-condition in \Cref{ln:dense}, the runtime is at most $O(nk/\alpha)$. Further, by \Cref{clm:bound-on-if-cond}, the condition on \Cref{ln:dense} holds at most $k/\alpha$ times with high probability. Since each time the algorithm enters the if-condition it spends $O(k)$ time to iterate over all elements, the total running time is at most $O(k^2/\alpha + nk/\alpha) = O(nk/\alpha)$.
\end{proof}

\begin{claim}\label{clm:alg3-time}
    \Cref{alg:sparsification-elements} runs in $O(k^2/\beta)$ time.
\end{claim}
\begin{proof}
    The algorithm samples $r_2 = O(k/\beta)$ sets and for each of them makes a membership query between the set and all elements. Hence, the total running time is upper-bounded by $O(k^2/\beta)$.
\end{proof}

We will prove the following lemma about the complexity of RGMM in the next section. However, to complete the time analysis, we state the lemma here.

\begin{restatable}{lemma}{rgmmlemma}\label{lem:rgmm}
    Let $H$ be the multigraph defined in \Cref{def:graph-h}. There exists an algorithm with an expected running time of $\wt{O}_\epsilon(k\beta + \alpha \beta n)$ that estimates the value of $E_{\pi} |\RGMM(H, \pi)|$ with $\epsilon k$ additive error with high probability.
\end{restatable}

\begin{lemma}\label{lem:total-time}
    The total running time of the algorithm is $\wt{O}(nk/\alpha + k^2/\beta + n\alpha\beta)$ with high probability.
\end{lemma}
\begin{proof}
    The algorithm runs in expected time $\wt{O}(nk/\alpha + n\alpha\beta)$ as shown by \Cref{clm:alg2-time}, \Cref{clm:alg3-time}, \Cref{lem:rgmm}, and the fact that $k \leq n$. To ensure a high-probability bound on the running time, we execute $\Theta(\log n)$ instances of the algorithm in parallel and use the estimate from the first instance that finishes. Since the expected running time is $\wt{O}(kn/\alpha + k^2/\beta + n\alpha\beta)$, the first instance is likely to terminate within $\wt{O}(nk/\alpha + k^2/\beta + n\alpha\beta)$ time with probability $1 - 1/\poly(n)$.
    (We remark that since our approximation guarantees hold with high probability, they also hold with high probability for each of the instances.)
\end{proof}

\subsection{Putting Everything Together}

We combine our ideas to obtain our final theorem for estimating $\chi = k - \SC(\mc{U}, \mc{F})$. Then, we extend our analysis and make slight modifications to the algorithm to estimate $k - \SC(\mc{U}, \mc{F}_{\neq 2})$ which is crucial for designing a sublinear algorithm for the Steiner tree problem.

\maintheorem*
\begin{proof}
    First, if $k \le O(n^{2/3})$, then we can easily query between all elements and sets and compute the estimate in $\wt{O}(n^{5/3})$ time since all the steps of the algorithm such as finding the maximal matching can be run in linear time with respect to the size of the input. Now suppose that $k > n^{2/3}$. We use \Cref{alg:set-cover} to produce the estimate $\wt\chi$, setting $x = 1/3$ and $y=1/3$. By \Cref{lem:approx-gaurantee}, we have that $\wt\chi$ is a multiplicative-additive $(1/2, \epsilon k)$-approximation to the value of $\chi$.

    Moreover, we have $\alpha = n^{1/3}$ and $\beta = 10 \max(k/n^{1-y}, 1) \cdot n \log(n) / k = \wt{O}(n^{1/3})$. By \Cref{lem:total-time}, the total runtime of the algorithm is upper-bounded by $\wt{O}(kn/\alpha + k^2/\beta + n\alpha \beta) = \wt{O}(n^{5/3})$ because of the assumption that $k \leq n$.
\end{proof}

Now we show how we can estimate $k - \SC(\mc{U}, \mc{F}_{\neq 2})$ with a slight modification. The only change that is needed in our algorithm is to remove edges of $H$ that are produced by sets of size 2 in $\hat{\mc{F}}$. Hence, when the algorithm estimates $\E_\pi[\RGMM(H, \pi)]$ using the vertex oracle, the oracles must not use those edges in the exploration as they do not exist in graph $H$. Thus, in the implementation of \Cref{lem:rgmm}, when the algorithm queries pairs $(u, S)$ where $u \in \mc{U}_{low}$ and $S \in \mc{F}_v$ to find neighbors of $v$ in $H$, if $u \in S$, the algorithm starts querying between all elements of $\mc{U} \setminus {v}$ and $S$ until it finds another element in $S$. If the algorithm finds another element of $S$, then it accepts the edge $(v, u)$ as a valid edge, and otherwise it continues the random search. So each time the algorithm finds such pair $(u, S)$, it invokes the above procedure to validate the edge. In the next theorem, we demonstrate how to bound the running time of the modified algorithm and discuss its approximation ratio.

\begin{theorem}\label{thm:setcover-without-two}
    There exists an algorithm that outputs a multiplicative-additive $(1/2, \epsilon k)$-approximation of value of $\chi = k - \SC(\mc{U}, \mc{F}_{\neq 2})$ in $\wt{O}(|\mc{F}|^{5/3})$ time with high probability.
\end{theorem}
\begin{proof}
    We need to show that the new estimation is a multiplicative-additive $(1/2, \epsilon k)$-approximation of $k - \SC(\mc{U}, \mc{F}_{\neq 2})$. We follow the same approach as proof of \Cref{lem:approx-gaurantee}. To follow the same steps, we need analogous claims similar to \Cref{clm:bound-on-mu} and \Cref{clm:bound-low-cover-to-all}. The exact same proof of \Cref{clm:bound-on-mu} also works to provide a bound on $\SC(\mc{U}_{low}, \hat{\mc{F}}_{\neq 2})$. More specifically, we have
    \begin{align*}
        \frac{1}{2}\left(|\mc{U}_{low}|- \SC(\mc{U}_{low}, \hat{\mc{F}}_{\neq 2}) \right) - \frac{\epsilon k}{2} \leq \wt\mu \leq |\mc{U}_{low}| - \SC(\mc{U}_{low}, \hat{\mc{F}}_{\neq 2}).
    \end{align*}
    Further, to show that $\SC(\mc{U}, \mc{F}_{\neq 2}) \leq \SC(\mc{U}_{low}, \hat{\mc{F}}_{\neq 2}) + \epsilon k/2$, note that we have $c = o(k)$ (see \Cref{alg:sparsification-sets} for definition of $c$) which implies that $\mc{U} \setminus (\mc{U}_{low} \cup \mc{U}_{high})$ can be covered using $o(k)$ sets of size larger than 2 (all sets that are removed in \Cref{alg:sparsification-sets} have size larger than 2). Additionally, elements of $\mc{U}_{high}$ can be covered using $\epsilon k /5$ sets of $\hat{\mc{F}}$ according to \Cref{clm:covering-high}. Thus, we can cover elements of $\mc{U}_{high}$ using $2\epsilon k /5$ sets of $\hat{\mc{F}}_{\neq 2}$ since we can replace each set of size 2 with two sets of size 1 that cover a single element. Finally, elements of $\mc{U}_{low}$ can be covered using $\SC(\mc{U}_{low}, \hat{\mc{F}}_{\neq 2})$ sets. The union of all these sets is a valid solution for $\SC(\mc{U}, \hat{\mc{F}}_{\neq 2})$ which has a size of at most $\SC(\mc{U}_{low}, \hat{\mc{F}}_{\neq 2}) + \epsilon k /2$. Therefore, the error of our estimation is $(1/2, \epsilon k)$.

    Now it remains to bound the running time of the modified algorithm. We prove that the same bound on the running time as in \Cref{lem:rgmm} holds. Now consider a vertex $v$ that is corresponding to an element in $\mc{U}_{low}$. Let $\hat{\mc{F}}_v$ be the collection of sets that include $v$. Also, let $\hat{\mc{F}}'_v = \{ S | S \in \hat{\mc{F}}_v \text{ and } |S| = 2 \}$. Let $r = |\hat{\mc{F}}_v|$ and $\hat{\mc{F}}_v = \{S_1, \ldots, S_r\}$. Also, let $\tau_i = |S_i \cap (\mc{U}_{low} \setminus v)|$. Hence, in \Cref{lem:rgmm}, when the algorithm queries a pair $(u, S)$ and it returns $u \in S$, the probability that $S = S_i$ is $\tau_i / \sum_{i=1}^r \tau_i$. For this set, using a Chernoff bound, the algorithm needs to spend at most $\wt{O}(k / \tau_i)$ time to see another element of $S_i$ or explore all elements of $\mc{U}$. Therefore, the expected additional time the algorithm needs to spend compared to \Cref{lem:rgmm} is
    \begin{align*}
       \sum_i^r \left(\frac{\tau_i}{ \sum_{i=1}^r \tau_i} \right) \cdot \wt{O}\left( \frac{k}{\tau_i}\right) &= \wt{O}\left(\sum_i^r \frac{k}{\sum_{i=1}^r \tau_i}\right) \leq \wt{O}\left(\frac{rk}{\deg_H(v)}\right),
    \end{align*}
    where the last inequality follows by the fact that $\deg_H(v) \leq \sum_{i=1}^r \tau_i$ since $\deg_H(v) = (\sum_{i=1}^r \tau_i) - |\{S | S \in \hat{\mc{F}}_v \text{ and } |S|=2\}|$ after the modification of the graph $H$. On the other hand, by \Cref{lem:outquery}, the expected number of times that the oracle calls an adjacent edge of vertex $v$, is at most $\wt{O}(\deg_H(v) / k)$ (note that we use $k$ here and $n$ in \Cref{sec:rgmm} are equivalent since in the next section we used $n$ as the number of vertices of a given graph). Also, these two variables have a negative correlation. Therefore, the total cost for all vertices is at most $\wt{O}(rk)$. Combining with the fact that $r = \wt{O}(\beta)$ (\Cref{lem:element-sparsification}), the total additional cost is $\wt{O}(k\beta)$, which is dominated by other terms in \Cref{lem:total-time} which implies that the algorithm has the same running time as \Cref{thm:setcover-general} if we choose $\alpha$ and $\beta$ similarly.
\end{proof}

\section{Random Greedy Maximal Matching on Multigraphs}\label{sec:rgmm}

In this section, we show that the algorithm of \cite{Behnezhad21} can be extended to work efficiently on multigraphs. First, we prove that if we have access to the adjacency list of a multigraph $H$, then it is possible to estimate the size of a random greedy maximal matching of $H$ in $\wt{O}(\bar{d})$ time, where $\bar{d}$ is the average degree of $H$. We denote $H = (V_H, E_H)$. We slightly abuse notation by using $n$ to denote $|V_H|$ (note that $n$ here is equivalent to $k$ in the set cover section). We also let $\pi$ represent the permutation over the edges of $H$ that we use to select edges of the RGMM.

Similar to the work of \cite{Behnezhad21}, we define a vertex oracle \VO{}, which, given a vertex $v \in V_H$ and a permutation $\pi$ over the edges of $H$, determines whether $v$ is matched in $\RGMM(H, \pi)$. This oracle explores the neighborhood around $v$ locally to answer the query and does not need to examine the entire graph. Additionally, to implement $\VO(v, \pi)$, we define an edge oracle \EO{}, which, given an edge $e \in E_H$ and a permutation $\pi$ over the edges of $E_H$, determines whether $e$ is matched in $\RGMM(H, \pi)$. Similar local oracles for random greedy maximal matching, random greedy maximal independent sets, and their modified versions have been extensively used in the literature \cite{Behnezhad21, BehnezhadRRS-SODA23, TSP-icalp24, NguyenOnakFOCS08, OnakSODA12, YoshidaYISTOC09}.

Both oracles are formally defined in \Cref{alg:vertexoracle} and \Cref{alg:edgeoracle}.

\begin{algorithm}[H]
\caption{Vertex oracle $\VO(v, \pi)$ to determine if vertex $v$ is matched in $\RGMM(H, \pi)$.}
\label{alg:vertexoracle}

	Let $e_1 = (v, u_1), \ldots, e_d = (v, u_d)$ be the edges incident to $v$ (including all copies of multiedges) with $\pi(e_1) < \ldots < \pi(e_d)$.
		
	\For{$i$ in $1 \ldots d$}{
		\If{$\EO(e_i, u_i, \pi) = \true$}{
                \Return \true
            }
	}
	\Return \false
\end{algorithm}

\begin{algorithm}[H]
\caption{Edge oracle $\EO(e, v, \pi)$ to determine an edge $e$ is in $\RGMM(H, \pi)$ where $v$ is an endpoint of $e$.}
\label{alg:edgeoracle}

    \lIf{$\EO(e, v, \pi)$ computed before}{\Return the computed result} \label{ln:cach}
    
	Let $e_1 = (v, u_1), \ldots, e_d = (v, u_d)$ be the edges incident to $e$ (including all copies of multiedges) such that $\pi(e_1) < \ldots < \pi(e_d) < \pi(e)$.
	
	\For{$i$ in $1 \ldots d$}{
		\If{$\EO(e_i, u_i, \pi) =  \true$}{\Return \false}
	}
	
	\Return \true
\end{algorithm}

It is not hard to see that the outputs of both the vertex oracle and the edge oracle are consistent with the RGMM, as the status of an edge (whether it is in the RGMM or not) depends only on the status of edges with a smaller rank in $\pi$. Both oracles utilize this property by querying the neighboring edges with smaller ranks in the permutation to determine if they are matched. Based on the results for these lower-ranked edges, the oracle can then decide if the vertex or edge will be matched.

\begin{observation}
    $\VO(v, \pi) = \true$ if and only if $v$ is matched in $\RGMM(H, \pi)$.
\end{observation}

\begin{observation}
    $\EO(e, v, \pi) = \true$ if and only if $e \in \RGMM(H, \pi)$.
\end{observation}

We define $T(v, \pi)$ as the number of recursive calls made by $\VO(v, \pi)$ for a vertex $v \in V_H$ and a permutation $\pi$ over $E_H$. The main result of this section is the following bound on $T(v, \pi)$.

\begin{theorem}\label{thm:multigraph-rgmm-bound-vertex}
    Let $\bar{d}$ be the average degree of $H$. It holds that $\E_{v, \pi}[T(v, \pi)] = O(\bar{d} \log n)$.
\end{theorem}

Let $Q(e, \pi)$ be the number of oracle calls to $\EO(e, \cdot, \pi)$ during the execution of $\VO$ for all vertices and permutation $\pi$. The following lemma is the main building block of the proof of \Cref{thm:multigraph-rgmm-bound-vertex}.

\begin{lemma}\label{lem:multigraph-rgmm-bound-edge}
    Let $e \in E_H$. It holds that $\E_{\pi}[Q(e, \pi)] = O(\log n)$.
\end{lemma}

Before proving \Cref{lem:multigraph-rgmm-bound-edge} we complete the proof of \Cref{thm:multigraph-rgmm-bound-vertex}. Recall that $\deg_H(v)$ denotes the degree of vertex $v$ in $H$, where multiedges are counted according to their number of occurrences.

\begin{proof}[Proof of \Cref{thm:multigraph-rgmm-bound-vertex}]
We have
\begin{align*}
    \E_{v, \pi}[T(v, \pi)] = \frac{1}{n}\E_\pi\Big[\sum_{v\in V_H}T(v,\pi)\Big] & = \frac{1}{n} \E_\pi\Big[\sum_{e \in E_H} Q(e, \pi)\Big] \\ 
    & = \frac{1}{2n} \sum_{v \in V_H} \sum_{e \ni v} \E_\pi[ Q(e, \pi)]\\
    & = \frac{1}{2n} \sum_{v \in V_H} \deg_H(v) \cdot O(\log n)\\
    & = O(\bar{d} \cdot \log n).
\end{align*}
\end{proof}

In what follows, we focus on proving \Cref{lem:multigraph-rgmm-bound-edge}. The general framework and approach we use to demonstrate this claim are similar to those in \cite{Behnezhad21}, with some modifications to accommodate multigraphs. We repeat all the steps to ensure completeness. For edge $e = (u, v)$, we use $\vec{e}$ to denote the directed edge from $u$ to $v$ and $\cev{e}$ for the directed edge from $v$ to $u$. Similarly, we define $Q(\vec{e}, \pi)$ (resp. $Q(\cev{e}, \pi)$) as the total number of queries to the edge $e$ directed from $u$ to $v$ (resp. from $v$ to $u$).

\begin{observation}\label{obs:direction-complxity}
    For any edge $e$, we have $\E_\pi[Q(e, \pi)] = \E_\pi[Q(\vec{e}, \pi)] + \E_\pi[Q(\cev{e}, \pi)]$.
\end{observation}
\begin{proof}
    The proof follows from the fact that $Q(e, \pi) = Q(\vec{e}, \pi) \cup Q(\cev{e}, \pi)$ and $Q(\vec{e}, \pi) \cap Q(\cev{e}, \pi) = \emptyset$.
\end{proof}

Let $\mc{R}$ be the set of edges stored in memory for recursive calls during the execution of the vertex oracle defined in \Cref{alg:vertexoracle}. Additionally, if the edge $e = (u, v)$ is visited with direction from $u$ to $v$ (i.e., $\EO(e, v, \pi)$ is called), we assume that $\vec{e}$ is stored in $\mc{R}$; otherwise, $\cev{e}$ is stored in $\mc{R}$. It is not hard to observe that at any point during the oracle calls, the edges stored in $\mc{R}$ have decreasing ranks according to the permutation $\pi$.

\begin{observation}\label{obs:decreasing-edges}
    Let $(e_1, e_2, \ldots, e_d)$ be edges that are stored in $\mc{R}$ during the execution of vertex oracle ($e_1$ stored first). Then, it holds that $\pi(e_1) > \pi(e_2) > \ldots > \pi(e_d)$.
\end{observation}
\begin{proof}
    The proof follows from the fact that $\EO(e, \cdot, \pi)$ only queries edges $e'$ for which $\pi(e') < \pi(e)$. Consequently, the edges stored in memory at any point are ranked in decreasing order.
\end{proof}

\begin{observation}
    Let $(e_1, e_2, \ldots, e_d)$ represent the edges stored in $\mc{R}$ at some point during the execution of the vertex oracle, in the order they are visited by the oracle. Then, $(e_1, e_2, \ldots, e_d)$ forms a path in $H$.
\end{observation}
\begin{proof}
    From the definition of the edge oracle, $P = (e_1, e_2, \ldots, e_d)$ is clearly a walk in the graph since each time the edge oracle queries an incident edge. Now we show that $P$ is a path. For the sake of contradiction, let $e_i$ be the smallest index where $(e_1, \ldots, e_{i-1})$ is a path but $(e_1, \ldots, e_{i})$ is not. Let $e_i = (w, z)$ where $w$ is shared with $e_{i-1}$. Hence, vertex $z$ is one of the endpoints of $e_1, \ldots, e_{i-1}$ (it might be $e_{i-1}$ since we have multiedges). Let $e_j$ ($j < i$) be the edge that the edge oracle queries as an incident edge of $z$ at the time it visits $z$. By \Cref{obs:decreasing-edges} we have that $\pi(e_i) < \pi(e_j)$. Thus, the edge oracle already queried $e_i$ before $e_j$ and the reason it did not add $e_i$ to the matching is because some other neighbor of $w$ with smaller ranking is already in the matching. Therefore, that neighbor must be queried before $e_i$ and the edge oracle must stop querying neighbors of $w$ at that point, which is a contradiction. Therefore, $P$ is a path.
\end{proof}

We define a query path as a path that is stored in $\mc{R}$ at some point during the execution of vertex oracle. Let $\vec{P} = (\vec{e_1}, \vec{e_2}, \ldots, \vec{e_d})$ be a query path for some permutation $\pi \in \Pi$, where $\Pi$ denotes the set of all permutations over the edges of multigraph $H$. Next, we define a function that maps a permutation $\pi$ to some permutation based on $\vec{P}$.

\begin{definition}[Function $f$]\label{def:function-f}
    Let $\vec{P} = (\vec{e_1}, \vec{e_2}, \ldots, \vec{e_d})$ be a query path and $\pi \in \Pi$ be some permutation over the edges of $H$. We define a function $f(\pi, \vec{P})\in \Pi$ that maps permutation $\pi$ to another permutation $f(\pi, \vec{P})$ such that $f(\pi, \vec{P})(e_i) = \pi(e_{i+1})$ for $i < k$, $f(\pi, \vec{P})(e_k) = \pi(e_1)$, and $f(\pi, \vec{P})(e) = \pi(e)$ for $e \notin \vec{P}$.
\end{definition}

Now based on \Cref{def:function-f}, we are able to construct a bipartite graph $H_{\Pi}^{\vec{e}} = (X^{\vec{e}}_{\Pi}, Y^{\vec{e}}_{\Pi}, E^{\vec{e}}_{\Pi})$ for a given edge $\vec{e}$ that is crucial for proving \Cref{lem:multigraph-rgmm-bound-edge}.

\begin{definition}[Bipartite Graph $H_{\Pi}^{\vec{e}} = (X^{\vec{e}}_{\Pi}, Y^{\vec{e}}_{\Pi}, E^{\vec{e}}_{\Pi})$]
    Let $\vec{e}$ be an edge in $H$ with some direction. We define $H_{\Pi}^{\vec{e}} = (X^{\vec{e}}_{\Pi}, Y^{\vec{e}}_{\Pi}, E^{\vec{e}}_{\Pi})$ to be a bipartite graph such that $|X^{\vec{e}}_{\Pi}| = |Y^{\vec{e}}_{\Pi}| = |E_H|!$, where each vertex in each part corresponds to a permutation over the edges of $H$. Let $x_{\pi} \in X_{\Pi}$ be a vertex that corresponds to permutation $\pi$ and $\vec{P}$ be a query path obtained on permutation $\pi$ for some vertex oracle call that ends to $\vec{e}$. Then, $E^{\vec{e}}_{\Pi}$ contains the edge $(x_{\pi}, y_{f(\pi, \vec{P})})$, where $y_{f(\pi, \vec{P})} \in Y^{\vec{e}}_{\Pi}$ corresponds to permutation $f(\pi, \vec{P})$.
\end{definition}

\begin{observation}\label{obs:deg-to-complexity}
    $\deg_{H_{\Pi}^{\vec{e}}}(x_{\pi}) = Q(\vec{e}, \pi)$.
\end{observation}
\begin{proof}
    For each query path that ends with $\vec{e}$ that can be produced by some vertex oracle call on permutation $\pi$, we add an incident edge to $x_{\pi}$. Hence, $\deg_{H_{\Pi}^{\vec{e}}}(x_{\pi}) = Q(\vec{e}, \pi)$.
\end{proof}

As a result of \Cref{obs:deg-to-complexity}, in order to prove \Cref{lem:multigraph-rgmm-bound-edge}, it is sufficient to prove that $\E_{x \sim X^{\vec{e}}_{\Pi}}[\deg_{H_{\Pi}^{\vec{e}}}(x)] = O(\log n)$. We call a permutation $\pi \in \Pi$ a {\em bad permutation} if and only if there exists a query path for this permutation that ends at $\vec{e}$ and it has a length larger than $c \log n$ for some large constant $c$. We use $\overline{X}^{\vec{e}}_{\Pi} \subset X^{\vec{e}}_{\Pi}$ to denote the set of vertices that correspond to bad permutations in $X^{\vec{e}}_{\Pi}$. We prove the following two lemmas, which are sufficient to achieve our goal.

\begin{lemma}\label{lem:small-long-perms}
    If $c$ is  a large enough constant, then $|\overline{X}^{\vec{e}}_{\Pi}| \leq |E_H|!/n^2$.
\end{lemma}

\begin{lemma}\label{lem:low-deg-likley}
    Let $y \in Y^{\vec{e}}_{\Pi}$. Then the number of neighbors of $y$ in $X^{\vec{e}}_{\Pi} \setminus \overline{X}^{\vec{e}}_{\Pi}$ is at most $c \log n$.
\end{lemma}

Their proofs are deferred to \Cref{sec:proof-of-small-long-perms,sec:proof-of-low-deg-likley}, respectively.

\begin{proof}[Proof of \Cref{lem:multigraph-rgmm-bound-edge}]
    First, note that $\deg_{H^{\vec{e}}_{\Pi}}(x_{\pi}) \leq O(n^2)$ for $x_{\pi} \in X^{\vec{e}}_{\Pi}$. To see this, suppose that we run $\VO(v, \pi)$ for some vertex $v$. The edge oracle for $e$ is either directly called by  $\VO(v, \pi)$ or from some neighboring edge oracle calls (the first time that the oracle visits the edge because of caching in \Cref{ln:cach} of \Cref{alg:edgeoracle}). Hence, $\VO(v, \pi)$ produces at most $n$ edge oracle calls to $e$. Since we have $n$ different options for $v$, the total number of query paths to $\vec{e}$ is at most $O(n^2)$. Further, each edge of $x_{\pi}$ in graph $H^{\vec{e}}_{\Pi}$ corresponds to a query path that ends with $\vec{e}$. Therefore, we have $\deg_{H^{\vec{e}}_{\Pi}}(x_{\pi}) \leq O(n^2)$. 
  
    Now, we prove that $\E_{x \sim X^{\vec{e}}_{\Pi}}[\deg_{H_{\Pi}^{\vec{e}}}(x)] = O(\log n)$. The total number of edges between $Y^{\vec{e}}_{\Pi}$ and $X^{\vec{e}}_{\Pi} \setminus \overline{X}^{\vec{e}}_{\Pi}$ is at most $|Y^{\vec{e}}_{\Pi}| \cdot c\log n$ by \Cref{lem:low-deg-likley}. Moreover, the total number of edges between $Y^{\vec{e}}_{\Pi}$ and $\overline{X}^{\vec{e}}_{\Pi}$ is at most $|\overline{X}^{\vec{e}}_{\Pi}| \cdot O(n^2)$ because $\deg_{H^{\vec{e}}_{\Pi}}(x_{\pi}) \leq O(n^2)$. Therefore, we have
    \begin{align*}
        |E^{\vec{e}}_{\Pi}| \leq |Y^{\vec{e}}_{\Pi}| \cdot c\log n + |\overline{X}^{\vec{e}}_{\Pi}| \cdot O(n^2) \leq O(|E_H|! \cdot  \log n),  
    \end{align*}
    where the last inequality follows by \Cref{lem:small-long-perms}. Thus
    \begin{align*}
        \E_{x \sim X^{\vec{e}}_{\Pi}}[\deg_{H_{\Pi}^{\vec{e}}}(x)] = \frac{|E^{\vec{e}}_{\Pi}|}{|E_H|!} \leq \frac{O(|E_H|! \cdot c \log n)}{|E_H|!} \leq O(\log n).
    \end{align*}
    Plugging in \Cref{obs:deg-to-complexity}, we obtain $\E_\pi[Q(\vec{e}, \pi)] = O(\log n)$. Finally, by \Cref{obs:direction-complxity}, we have $\E_\pi[Q(e, \pi)] = O(\log n)$.
\end{proof}

\subsection{Implementation Details of \RGMM{} for Multigraphs}

Until now, we demonstrate that the query complexity of our vertex oracle is $O(\bar{d} \log n)$. Now we show how to utilize this oracle to estimate the size of \RGMM{}. The following result is due to \cite{Behnezhad21}.

\begin{proposition}[See Appendix A in \cite{Behnezhad21}]\label{lem:reduction-cost-neighbor}
    Let $T(v)$ be the time needed to return a random neighbor of vertex $v$ that is not exposed to $v$ yet. Also, let $Q(v)$ be the expected number of times that the algorithm needs a random neighbor of $v$ if we start the oracle calls from a random vertex for a random permutation. Then, the expected time to run the vertex oracle for a random vertex and a random permutation is at most $\sum_v T(v) \cdot Q(v)$.
\end{proposition}

\begin{remark}
    Note that \cite{Behnezhad21} proved the above lemma for $T(v) = O(1)$, but essentially the same proof applies to any $T(v)$.
\end{remark}

Up to this point, we are able to estimate $E_{\pi} |\RGMM(H, \pi)|$ (using $\wt{\Theta}(1)$ uniformly random vertex oracle calls) efficiently if we have access to the adjacency list of the multigraph $H$ because of \Cref{lem:reduction-cost-neighbor}. However, this is not feasible in our application because we do not have access to the adjacency list of $H$, and building it before running the vertex oracle is too costly. We use the following property of \RGMM{} to refine \Cref{lem:reduction-cost-neighbor} into a tighter bound that can be applied to our set cover problem.

\begin{lemma}\label{lem:outquery}
    Let $Q(v)$ be the expected number of times that the oracle queries an adjacent edge of $v$ if we start the oracle calls from a random vertex, for a random permutation over the edges of the multigraph $H$. It holds that $Q(v) = \wt{O}(\deg_H(v) / n)$.
\end{lemma}

\begin{proof}
    Let $e_1, e_2, \ldots, e_{\deg_H(v)}$ be all edges incident to $v$ in $H$. Thus, we have 
    \begin{align*}
        Q(v) = \frac{1}{n} \cdot \frac{1}{|E_H|!}\sum_{\pi \in \Pi}\sum_{i=1}^{\deg_H(v)} Q(e_i, \pi) &= \frac{1}{n}\sum_{i=1}^{\deg_H(v)} \frac{1}{|E_H|!} \sum_{\pi \in \Pi} Q(e_i, \pi)\\
        & = \frac{1}{n}\sum_{i=1}^{\deg_H(v)} \E_\pi[Q(e_i, \pi)]\\
        & = \frac{1}{n}\sum_{i=1}^{\deg_H(v)}O(\log n) &(\text{By \Cref{lem:multigraph-rgmm-bound-edge}})\\
        & = \wt{O}(\deg_H(v) / n)
    \end{align*}
\end{proof}

\begin{corollary}\label{cor:reduction-cost-neighbor}
    Let $T(v)$ be the time needed to return a random neighbor of vertex $v$ that is not exposed to $v$ yet. Then, the expected time to run the vertex oracle for a random vertex and a random permutation is $\sum_v \wt{O}(T(v) \cdot \deg_H(v) / n)$.
\end{corollary}
\begin{proof}
    The proof can be obtained by plugging \Cref{lem:outquery} into \Cref{lem:reduction-cost-neighbor}.
\end{proof}

\subsection{Proof of \Cref{lem:small-long-perms}} \label{sec:proof-of-small-long-perms}

We use the following result for the round-complexity of the parallel randomized greedy maximal independent set by \cite{FischerN18}.

\begin{proposition}[\cite{FischerN18}]\label{pro:round-complexity-mis}
    Let $\pi$ be a permutation of the vertices of a graph $G$ with $n$ vertices, drawn uniformly at random. For any constant $c$, with probability $1 - n^{-c}$, we have $\rho(G, \pi) = O(\log n)$.
\end{proposition}

Given a graph $G$, we can construct its line graph $L(G)$ by adding a vertex in $L(G)$ for each edge of $G$ and adding an edge between two vertices of $L(G)$ if their corresponding edges share an endpoint in $G$. It is easy to see that a random greedy MIS on $L(G)$ is equivalent to a random greedy maximal matching of $G$, which implies the following corollary as a result of \Cref{pro:round-complexity-mis}.

\begin{corollary}\label{cor:round-complexity-mis}
    Let $\pi$ be a permutation of the edges of a graph $G$ with $n$ vertices, drawn uniformly at random. For any constant $c$, with probability $1 - n^{-c}$, we have $\rho(L(G), \pi) = O(\log n)$.
\end{corollary}

Note that for any edge that appears in the solution of parallel random greedy MIS of $L(H)$, if the edge oracle queries that edge, the answer to this query is going to be consistent since they simulate the same greedy algorithm with respect to the given permutation. Now we prove that the round-complexity of random greedy MIS on the line graph of $H$ is large for bad permutations which is enough to show \Cref{lem:small-long-perms}. Let $\pi$ be a bad permutation. Hence, there exists a query path $\vec{P} = (\vec{e_1}, \vec{e_2}, \ldots, \vec{e_r})$ for this permutation such that $r > c \log n$. Let $\rho(e_i)$ be the round edge $e_i$ is removed from $L(H)$ when we run parallel randomized greedy MIS on $L(H)$ with respect to $\pi$. We claim that $\rho(e_i) > \rho(e_{i+2})$ for $1 < i < r - 1$. For the sake of contradiction, assume that $\rho(e_i) \leq \rho(e_{i+2})$. Thus, $\rho(e_i) \leq \rho(e_{i+1})$ since if $e_{i+1}$ removed from the graph before $e_i$, it means that it is removed because an edge adjacent to $e_{i+1}$ but not adjacent to $e_i$ is in MIS of $L(H)$ which implies that $e_{i+2}$ must be removed at the same or before, so $\rho(e_{i+2}) \leq \rho(e_{i+1}) < \rho(e_i)$ which is a contradiction. So it must hold that $\rho(e_i) \leq \rho(e_{i+1})$ which implies that either $e_i$ and $e_{i+1}$ are getting deleted from $L(H)$ in the same round or $e_i$ is removed before $e_{i+1}$. Hence, when we are removing $e_i$, it is not a local minimum with respect to $\pi$. Now consider the round that $e_i$ is removed from $L(H)$. Let $e_i = (u, v)$ and $v$ be the shared endpoint with $e_{i+1}$. There are three possible scenarios:
\begin{itemize}
    \item \textbf{$e_i$ is removed because another multiedge $(u,v)$ is in MIS:} Let $e' = (u, v)$ be the multiedge that is a local minimum in the round we remove $e_i$. Note that we have $\pi(e') < \pi(e_i)$. Thus, the edge oracle of $e_{i-1}$ must first queries $e'$ before $e_i$. Then, since $e'$ is in random greedy MIS of $L(H)$ it must stop the process at this point and does not query $e_i$ which is a contradiction.

    \item \textbf{$e_i$ is removed because edge $(u, w)$ is in MIS where $w$ some vertex of the graph:} Let $e' = (u, w)$. Similar to the previous case, the edge oracle for $e_{i-1}$ must first query $e'$ and since $e'$ is in the solution it must stop the process which is a contradiction.

    \item \textbf{$e_i$ is removed because edge $(v, w)$ is in MIS where $w$ some vertex of the graph:} Let $e' = (v, w)$. Similar to the previous case, the edge oracle for $e_i$ must first query $e'$ and since $e'$ is in the solution it must stop the process which is a contradiction.
\end{itemize}

Therefore, we have $\rho(e_i) > \rho(e_{i+2})$ which implies that $\rho(e_2) \geq \rho(e_4) + 1 \geq \rho(e_6) + 2 \geq \ldots \geq r/2 - 1$. Let $c$ be sufficiently large enough. Hence, we have $\rho(L(H), \pi) \geq r/2 - 1 \geq c\log n / 2 - 1$. By \Cref{cor:round-complexity-mis}, this situation occurs for at most a $1/n^2$ fraction of permutations, which completes the proof of \Cref{lem:small-long-perms}.

\subsection{Proof of \Cref{lem:low-deg-likley}} \label{sec:proof-of-low-deg-likley}

Let $y_\pi \in Y^{\vec{e}}_{\Pi}$. Also, let $x_{\pi_1}, x_{\pi_2} \in X^{\vec{e}}_{\Pi} \setminus \overline{X}^{\vec{e}}_{\Pi}$ be two vertices corresponding to permutations $\pi_1$ and $\pi_2$ such that they are connected to $y_\pi$. Let $\vec{P}_1$ and $\vec{P}_2$ be the two query paths such that $f(\pi_1, \vec{P}_1) = f(\pi_2, \vec{P}_2) = \pi$. We show that either $P_1 \subseteq P_2$ or $P_2 \subseteq P_1$. In other words, one of $P_1$ or $P_2$ must be a subpath of the other path. This is enough to complete the proof of \Cref{lem:low-deg-likley} the longest query path of all permutations in $X^{\vec{e}}_{\Pi}$ is $O(\log n)$ by the definition. So for the rest of this subsection, we focus on proving this claim.

Let $\vec{P}_1 = (\vec{e_{k_1}}, \ldots , \vec{e_2},\vec{e_1})$ and $\vec{P}_2 = (\vec{e_{k_2}}', \ldots , \vec{e_2}',\vec{e_1}')$ where we have $e_1 = e_1'=e$. For the sake of contradiction, suppose that none of the two paths is a subpath of the other one.  Because of this assumption, there exists some $i$ such that $\vec{e_j} = \vec{e_j}'$ for $j \leq i$ and $\vec{e_{i+1}} \neq \vec{e_{i+1}}'$. Without loss of generality, assume that $\pi_1(e_i) \leq \pi_2(e_i)$ (because of the symmetry up to this point in the proof). By the definition of function $f$ (\Cref{def:function-f}) and the fact that $\pi = f(\pi_1, \vec{P}_1) = f(\pi_2, \vec{P}_2)$, we have
\begin{align}
    \pi_2(e_{i+1}) = f(\pi_2, \vec{P}_2)(e_{i+1}) = f(\pi_1, \vec{P}_1)(e_{i+1}) = \pi_1(e_i). \label{eq:equal-position-perm}
\end{align}
Moreover, combining with equality \Cref{eq:equal-position-perm} and the assumption that $\pi_1(e_i) \leq \pi_2(e_i)$, we get
\begin{align*}
    \pi_2(e_{i+1}) = \pi_1(e_{i}) \leq \pi_2(e_i).
\end{align*}
On the other hand, we know $\pi_2$ is a permutation over edges of $H$ which implies that
\begin{align}
    \pi_2(e_{i+1}) < \pi_2(e_i). \label{ln:smaller-edge}
\end{align}
Let $\hat{e}$ be an arbitrary edge. We claim that if $\min(\pi_1(\hat{e}), \pi_2(\hat{e})) < \pi_1(e_i)$, then $\pi_1(\hat{e}) = \pi_2(\hat{e})$. If $\hat{e} \notin P_1 \cup P_2$, then the claim clearly holds. If $\hat{e} \in \{e_1, \ldots, e_{i-1}\}$, i.e. $\hat{e} = e_j$ for $j <i$, we have
\begin{align*}
    \pi_1(\hat{e}) = \pi_1(e_j) = f(\pi_1, \vec{P}_1)(e_{j+1}) = f(\pi_2, \vec{P}_2)(e_{j+1}) = \pi_2(e_j) = \pi_2(\hat{e}).
\end{align*}
In all other cases, $\min(\pi_1(\hat{e}), \pi_2(\hat{e})) \geq \pi_1(e_i)$.

\begin{claim}\label{clm:ei1-in-rgmm}
    Assuming that neither of $P_1$ or $P_2$ is a subpath of the other, the edge $e_{i+1}$ is in the random greedy maximal matching of $H$ with respect to permutation $\pi_2$.
\end{claim}
\begin{proof}
    If $e_{i+1}$ is not in the \RGMM{} of permutation $\pi_2$, there must exists some edge $\hat{e}$ that blocks $e_{i+1}$ from being in \RGMM{} such that $\pi_2(\hat{e}) < \pi_2(e_{i+1})$. Combining with equality \Cref{eq:equal-position-perm}, we obtain $\pi_2(\hat{e}) < \pi_1(e_i)$. Since $\min(\pi_1(\hat{e}), \pi_2(\hat{e})) < \pi_1(e_i)$, we have $\pi_1(\hat{e}) = \pi_2(\hat{e})$. Thus, $\pi_1(\hat{e}) < \pi_1(e_i)$. Because both $\pi_1$ and $\pi_2$ are similar up to ranking $\pi(e_i)$, then edge $\hat{e}$ must also be in \RGMM{} of permutation $\pi_1$.

    Let $\vec{e_{i+1}} = (u, v)$. There are three possible scenarios for $\hat{e}$:
    \begin{itemize}
        \item \textbf{$\hat{e} = (u, v)$, i.e. is one of the multiedges between $u$ and $v$ similar to $e_{i+1}$:} In this case, either $\VO(u, \pi_1)$ or $\EO(e_{i+2}, u, \pi_1)$ query $\EO(\hat{e}, v, \pi_1)$ before $\EO(e_{i+1}, v, \pi_1)$ since $\pi_1(\hat{e}) < \pi_1(e_{i+1})$. Since $\hat{e}$ is in \RGMM{} of permutation $\pi_1$, the process terminates and $\vec{P}_1$ is not a valid query path.
 
        \item \textbf{$\hat{e} = (u, w)$ for some vertex $w$:} In this case, either $\VO(u, \pi_1)$ or $\EO(e_{i+2}, u, \pi_1)$ query $\EO(\hat{e}, w, \pi_1)$ before $\EO(e_{i+1}, v, \pi_1)$ since $\pi_1(\hat{e}) < \pi_1(e_{i+1})$. Since $\hat{e}$ is in \RGMM{} of permutation $\pi_1$, the process terminates and $\vec{P}_1$ is not a valid query path.

        \item \textbf{$\hat{e} = (v, w)$ for some vertex $w$:} In this case, either $\EO(e_{i+1}, v, \pi_1)$ query $\EO(\hat{e}, w, \pi_1)$ before $\EO(e_{i}, \cdot, \pi_1)$ since $\pi_1(\hat{e}) < \pi_1(e_i)$. Since $\hat{e}$ is in \RGMM{} of permutation $\pi_1$, the process terminates and $\vec{P}_1$ is not a valid query path.
    \end{itemize}
Therefore, $e_{i+1}$ is in the random greedy maximal matching of $H$ with respect to permutation $\pi_2$.
\end{proof}

Now we are ready to prove the contradiction which completes the proof of \Cref{lem:low-deg-likley}. By inequality \Cref{ln:smaller-edge}, edge oracle $\EO(e_{i+1}', \cdot, \pi_2)$ must query $e_{i+1}$ before $e_{i}$. Also, by \Cref{clm:ei1-in-rgmm}, edge $e_{i+1}$ is in the random greedy maximal matching of $H$ with respect to permutation $\pi_2$, which means that $P_2$ is not a valid query path.

\subsection{Implementation of RGMM for Set Cover with Our Access Model} \label{sec:implementation-of-rgmm}

In this section, we demonstrate how we can change the \RGMM{} algorithm to work with our access model in the set cover problem instead of having access to the adjacency list. The challenge arises when the algorithm needs to find a random neighbor of a vertex 
$v$ (which corresponds to an element in the set cover problem) in graph $H$. Naively, the algorithm queries all sets in $\hat{\mc{S}}$ and finds all sets that cover the element corresponding to $v$. Then, for each of the sets, it queries all elements in $\mc{U}_{low}$ to find all neighbors of $v$ in $H$.

\rgmmlemma*

\begin{proof}
    First, it is important to mention that the number of vertices in $H$ is $k$. In order to estimate $E_{\pi} |\RGMM(H, \pi)|$ with $\epsilon k$ additive error, it is sufficient to run the vertex oracle for $\Theta(1/\epsilon \cdot \poly \log k))$ random vertices and permutations. Then, using a Chernoff bound, it is easy to show that we can estimate  $E_{\pi} |\RGMM(H, \pi)|$ with $\epsilon k$ additive error. 

    In order to simulate the oracles in our access model and use \Cref{cor:reduction-cost-neighbor}, we demonstrate how we can find a random neighbor of a vertex $v$ in $H$. Consider the first time that the algorithm needs to find a random neighbor of $v$. We first query all sets in $\hat{\mc{S}}$ to identify those that contain the element corresponding to $v$. Let $\hat{\mc{S}}_v$ be the collection of sets that include $v$. This step takes $O(n)$ time since $|\hat{\mc{S}}| \leq n$. By \Cref{lem:element-sparsification}, we have $|\hat{\mc{S}}_v| \leq \wt{O}(\beta)$ with high probability. Now consider all pairs of $(u, S)$ where $v \in \mc{U}_{low}$ and $S \in \hat{\mc{S}}_v$. There are at most $\wt{O}(k\beta)$ such pairs. We start to query these pairs randomly until finding an element that exists in one of the sets of $\hat{\mc{S}}_v$. Note that in expectation, the algorithm needs to make $T(v) = \wt{O}(k\beta / \deg_H(v))$ queries to find such an element which is a neighbor of $v$ in $H$. Therefore, using a concentration inequality such as a Chernoff bound, with high probability the algorithm finds such an element in $T(v) \cdot \poly \log n = \wt{O}(k\beta / \deg_H(v))$ attempts. Also, since the algorithm makes all the queries in random order, it has the same probability of seeing any edges adjacent to $v$.

    Therefore, each time that the algorithm needs to find a neighbor of vertex $v$ in $H$, it spends at most $\wt{O}(n + k\beta / \deg_H(v))$ time. Now by \Cref{cor:reduction-cost-neighbor}, the expected time to run the vertex oracle for a random vertex and a random permutation is
    \begin{align*}
        \sum_v \wt{O}\left(\frac{T(v) \cdot \deg_H(v)}{k}\right) 
        = \sum_v \wt{O}\left(\frac{(n + k\beta / \deg_H(v)) \cdot \deg_H(v)}{k}\right)
        & = \sum_v \wt{O}\left( \frac{n\cdot \deg_H(v)}{k} + \beta \right)\\
        & = \wt{O}\left(k\beta + n\cdot \sum_v \frac{\deg_H(v)}{k}\right)\\
        & = \wt{O}(k\beta + n \bar{d}),
    \end{align*}
    where $\bar{d}$ is the average degree of multigraph $H$.

    On the other hand, each set in $\hat{\mc{S}}_v$ has at most $\wt{O}(\alpha)$ neighbors in $\mc{U}_{low}$ by \Cref{lem:set-sparsification-gaur}. Therefore, vertex $v$ has at most $\wt{O}(\alpha \cdot \beta)$ neighbors in $H$ since $|\hat{\mc{S}}_v| \leq \wt{O}(\beta)$. Thus, we have $\bar{d} \leq \wt{O}(\alpha \cdot \beta)$ which completes the proof. 
\end{proof}

\section{Connection to Steiner Tree}\label{sec:steiner-tree}
In this section, we show how our improved algorithm for set cover (from Section~\ref{sec:set-cover}) implies an improved sublinear algorithm for metric Steiner tree. Formally, we show the following theorem.

\maintheoremsteiner*

The overall structure of our algorithm is similar to the algorithm of Chen, Khanna and Tan~\cite{chen2023query}. The main difference in our algorithm compared to~\cite{chen2023query} is in the set cover component.
In the following, we first provide an overview of their algorithm in \Cref{sec:steiner-highlevel}. We then provide the query-efficient implementation of their algorithm and our modification to it in \cref{sec:setp-2-query}. We will finally provide the query complexity analysis and the proof of \cref{thm:steiner-tree} in \cref{sec:steiner-analysis}. Note that the approximation analysis of our algorithm follows directly from the proof of Theorem 3 in~\cite{chen2023query}.

\subsection{Algorithm at a High Level}\label{sec:steiner-highlevel}
\paragraph{Step 1: Minimum spanning tree over terminals.}
The algorithm of~\cite{chen2023query} as a first step starts with an MST $\mathcal{T}^*$ over the terminals $T$ (whose cost can be estimated in nearly linear time using the sublinear MST algorithm of Czumaj and Sohler~\cite{czumaj2009estimating}).
It is known that $w(\mathcal{T}^*) /2 \le \ST(V,T,w) \le w(\mathcal{T}^*)$. To get a strictly better-than-$2$ approximation of $\ST(V,T,w)$, it suffices to detect whether $\ST(V, T, w)$ is closer to $w(\mathcal{T}^*)/2$ or $w(\mathcal{T}^*)$. Hence, the rest of the algorithm is either to provide ``significant'' local improvements over $w(\mathcal{T}^*)$ using ``set cover'' like structure (i.e., step 2 in Section 4 of~\cite{chen2023query}) or ``local structure'' (i.e., step 3 in Section 4 of~\cite{chen2023query}), and thus output $(1-O(\eta)) w(\mathcal{T}^*)$ as the estimate of $\ST(V, T, w)$; or conclude that $\ST(V, T, w)$ is closer to $w(\mathcal{T}^*)$ and output it as the estimate of $\ST(V, T, w)$. Then, they show how to implement these steps using sublinear queries to~$\mathcal{O}$.

\paragraph{Step 2: Improvement using set cover.} 
First, they partition the edges  into $L=O((\log k)/\varepsilon)$ buckets such that the edges of the $i$th bucket have weights in $[(1+\varepsilon)^{i-1}, (1+\varepsilon)^{i})$. Let $H_i$ be the graph built on all the terminals and all the edges upto the $i$th bucket. They define an instance of set cover corresponding to each level $i$ where \emph{ideally}, 
\begin{itemize}
    \item The elements correspond to the connected components of $H_{i-1}$.
    \item The sets correspond to the Steiner vertices.
    \item A set $W_v$ (corresponding to a Steiner vertex $v$) contains an element $u_S$ (corresponding to a component $S$) if the distance between $v$ and some terminal in $S$ is less than a threshold $\tau$ (think of it as $\frac{3}{5} \cdot (1+\varepsilon)^i$). 
\end{itemize}

\paragraph{How to use set cover in making a decision about the Steiner tree cost.} They show that if one can solve each of these set cover instances approximately, then one can check whether the total contribution of these set cover improvements is more than $O(\eta)\cdot w(\mathcal{T}^*)$. In particular, in that case $\ST(V, T, w)$ is strictly less than $(1-O(\eta)) \cdot w(\mathcal{T}^*)$. Intuitively, this is because one can include the Steiner vertices corresponding to the set cover solution and remove a subset of the edges in $\mathcal{T}^*$, while still maintaining a feasible solution to the Steiner tree instance. Hence, this implies that the cost of the constructed solution of the Steiner tree instance  is less than $(1-O(\eta)) \cdot w(\mathcal{T}^*)$.

\paragraph{A challenge and the notion of representatives.} The main challenge with the above algorithm is computing the set cover instance. More precisely, in the third bullet point above, in order to check whether a set (corresponding to a Steiner vertex $v$) contains an element (corresponding to a connected component $S$), they need to compute the distance of $v$ to all terminals $t\in T$ which could be very costly. Instead, they define a \emph{net} on $S$ which is a maximal subset $\tilde{S}\subseteq S$ such that any pair of terminals in $\tilde{S}$ has distance at least $\varepsilon \cdot (1+\varepsilon)^i$. They call the terminals in $\tilde{S}$ the \emph{representatives}. 

\paragraph{Modified set cover instance and the notion of light/heavy levels.} Now, to detect if a set contains an element, we only need to check the distance of $v$ to all terminals in $\tilde{S}$. When $|\tilde{S}|$ is small, this can be done efficiently. So, in their algorithm they \emph{only} assign an element to a connected component $S$ if the size of its representatives is \emph{small}. To show that this does not introduce a large error, they define a level $i$ to be \emph{light} if the total sum of the edges of $\mathcal{T}^*$ in bucket $i$ is ``small'', and define it to be \emph{heavy} otherwise. They show that one can ignore all the levels $i$ that are light, and moreover if a level is heavy, then most of its components have small sets of representatives and thus the error introduced by ignoring the components $S$ with large net size is negligible.

Finally, we note that the notion of representatives will be used in other parts of the overall algorithm such as computing $\mathcal{T}^*$ which we will go over when describing the implementation of this step.

\paragraph{Step 3: Improvement using ``local structure''.}
In this step, they consider the hierarchical structure of the connected components. Specifically, they focus on components $S$ that have exactly two child components, $S_1$ and $S_2$, where each of these child components also has exactly two child components: $S_{11}$, $S_{12}$ for $S_1$ and $S_{21}$, $S_{22}$ for $S_2$. Then, they check whether there exists a single Steiner vertex $v$ that can be used to connect components $S_{11}, S_{12}, S_{21}, S_{22}$ and instead remove the corresponding edges in $\mathcal{T}^*$ connecting these components. Similarly to step 2, if the overall advantage of all these 2-level local improvements is more than $O(\eta)\cdot w(\mathcal{T}^*)$, then the algorithm outputs $(1-O(\eta))\cdot w(\mathcal{T}^*)$ as its estimate of $\ST(V, T,w)$.

Given that our algorithm does not change this step at all, we refer the reader to~\cite{chen2023query} for further details. Moreover, the query complexity is exactly the same as in~\cite{chen2023query}.

\subsection{Implementation of the Algorithm}\label{sec:setp-2-query}
Here, we focus on a query-efficient implementation of the algorithm, and particularly highlight where the set cover component was used and how we modify it.

First, we note that one cannot compute $H_i$ exactly, so~\cite{chen2023query} shows that it suffices to work with an approximate graph $H'_i$ such that $H_i \subseteq H'_{i} \subseteq H_{i+1}$. 

\paragraph{Subroutines.} Next, \cite{chen2023query}~define some useful subroutines for simulating the set cover instance on $H'_i$:
\begin{itemize}
    \item{{\sf Find}($u,i$):} This receives a terminal $u$ and a level $i$, and finds some representative terminal $u'$ in the same connected component of $H'_i$ that contains $u$. Moreover, they show that this subroutine can be implemented using $\wt{O}(k)$ queries.
    \item{{\sf BFS}($u,i$):} This subroutine reports all representative terminals that are in the same connected component in $H'_i$ as $u$, if the total number of such representatives is $\wt{O}(L/\varepsilon) = \wt{O}(1/\varepsilon)$. Otherwise, the procedure is terminated. This procedure employs {\sf Find} subroutines and has total query complexity of $\wt{O}(k/\varepsilon)$ which is $\wt{O}(k)$ given that $\varepsilon$ is a constant.
\end{itemize}

\paragraph{Parameters.} The algorithm uses four parameters, whose values will be later set to optimize the query complexity of the algorithm.

\begin{itemize}
    \item $M$ is a threshold parameter used on the number of components that contain a ``small'' number of representatives. Note that these are the components to which we assign an element in the universe $\sU_i$ of the  corresponding set cover instance $(\sU_i, \sF_i)$.
    \item $R$ is a threshold parameter defining ``low-degree'' and ''high-degree'' elements.
    \item $P$ is a threshold parameter denoting ``low-degree'' or ``high-degree'' sets.
    \item $\kappa$ is a threshold parameter on the value of $k$. At a high level, when $k<\kappa$, we can afford to query all distances between terminals and the Steiner vertices. 
\end{itemize}

\paragraph{Simulation of Step 2 and its query complexity.}
We now outline the implementation of Step 2 and specify the query complexity of each step, along with potential conditions they impose on the parameters we need to set.

\begin{itemize}[leftmargin=*]
    \item {\bf The case of small number of terminals: if} $k \le \kappa$, then we query all distances between terminals and Steiner vertices, which takes $O(n\kappa)$ queries. Then, we estimate $|\sU| - \SC(\sU, \sF)$ using our algorithm, but without any further queries.
    \item {\bf Otherwise, } for each level $i$, they show that one of the following cases hold:  
        
        \paragraph{Case 1:} The total number of representative terminals in all connected components is $\wt{O}(M/\varepsilon)$. 
        To detect this case, they use greedy MIS which can be implemented by the {\sf BFS} and {\sf Find} subroutines and will take $\wt{O}(Mk/\varepsilon)$  queries (for further details, see~\cite{YoshidaYISTOC09}). If this is the case, then again the set cover instance can easily be computed by querying the distance of all Steiner vertices to all representative terminals which requires $\wt{O}(nM/\varepsilon)$ queries. Then, similarly to the case of $k \le \kappa$, $|\sU_i|-\SC(\sU_i, \sF_i)$ can be estimated using our algorithm, without any further queries.

        \paragraph{Case 2:} $|\sU_i|\leq M$. In this case they show that level $i$ is in fact light and thus can be ignored. To detect this case, they estimate the size of $|\sU_i|$ using calls to BFS starting from $\wt{O}(k/M)$ random terminals, which overall takes $\tilde O(k^2/M)$ queries. 
        Note that   if we are in this case, we take no further action. Also, this step requires $k > M$, which we will ensure in our parameter setup.

        \paragraph{Case 3:} The last case is when $|\sU_i|\geq M$. 
        \begin{itemize}[leftmargin=*]
            \item {\bf Partitioning of the terminals based on their degree.} First, they partition the terminals into $T_{low}$ and $T_{high}$ based on whether the number of ``close-by'' Steiner vertices (roughly within distance $(3/5)(1+\varepsilon)^i$) to them is smaller than or larger than $R$. This partitioning can be computed using $\wt{O}(kn/R)$ queries by randomly sampling $\wt{O}(n/R)$ Steiner vertices and checking their distance to all the terminals.
            \item {\bf Handling high-degree terminals.} Then, by picking $\wt{O}(n/R)$ sets uniformly at random, with high probability, all elements corresponding to the components containing at least one terminal in $T_{high}$ are covered. As they can only afford an $\varepsilon |\sU_i|$ additive error in their estimate of $|\sU_i| - \SC(\sU_i, \sF_i)$, they require that $n /R < \wt{O}(\varepsilon M) = \wt{O}(\varepsilon |\sU_i|)$.   
            \item {\bf Handling low-degree terminals.} Next, they solve the set cover instance on $\sU_{low}$, i.e., the connected components that have no terminal in $T_{high}$.
            \begin{itemize}[label=$\bullet$, leftmargin=*]
                \item {\bf Partitioning of the Steiner vertices based on their degree.} They partition the sets of $\sF_i$ into $\sW_1$ and $\sW_2$ based on whether their degree to $T_{low}$ is less than $\Theta(P)$ or higher. This partitioning can be computed using $\wt{O}(nk/P)$ queries by randomly sampling $k/P$ terminals from $T_{low}$. This requires $k> \wt{\Omega}(P)$ which we will ensure in our parameter setup. Then, they consider the set cover instances $(\sU_{low},\sW_1)$ and $(\sU_{low}, \sW_2)$ separately, and return the better of the two solutions.
                
                \medskip
                {\bf Our modification to this step:} In our set cover algorithm, however, we define $\sW_2$ slightly differently, as described in Set Sparsification, see~\Cref{alg:sparsification-sets}. 
                More precisely, we iterate over Steiner vertices (i.e., sets in $\sF_i$) one by one in an arbitrary order, and at every round $j \le |\sF_i|$, we check whether the degree of the Steiner vertex $v_j$ to $T_{low}$ is more than $\Theta(P)$. If so, we add its corresponding set, $W_j$, to $\sW_2$. 
                Similarly to their test, our test can also be implemented using $\wt{O}(nk/P)$ queries. However, each time we add a set $W_j$ to $\sW_2$, we find all its ``nearby'' terminals and remove them from $T_{low}$, more precisely, $T_{low} \leftarrow T_{low} \setminus W_j$. This step can be simply done by querying the distance of the Steiner vertex $v_j$ and all terminals in $T_{low}$. Hence, each time a set is added to $\sW_2$, we perform an extra $\wt{O}(k)$ queries compared to the algorithm of~\cite{chen2023query}. However, as we can simply bound $|\sW_2|$ by $k/P$, the overall query complexity remains as $\wt{O}(nk/P + k^2/P) = \wt{O}(nk/P)$.  
                
                \medskip
                \item {\bf Handling high-degree Steiner vertices.} To solve $(\sU_{low},\sW_2)$, note that by a simple double-counting argument, $|\sW_2|\leq kR/P$. So if $kR/P\leq \varepsilon M\leq \varepsilon |\sU_i|$, which will be ensured in the parameter setting, we can afford to pick all sets in $\sW_2$ and thus, similarly to their argument, we only need to estimate $|\bigcup_{W\in \sW_2} W|$ in the set cover instance. This is done by randomly sampling the terminals and using BFS and will take an overall $\wt{O}(k^2 / M)$ queries. 
                
                \medskip
                {\bf Our modification to this step:} With the adjusted partitioning of $\sF_i$ into $\sW_1$ and $\sW_2$ in our algorithm, the size of $\sW_2$ is at most $k/P$. Therefore, by setting $k/P \leq \varepsilon M \le \varepsilon |\sU_i|$, we can afford to select all sets in $\sW_2$. Notably, this modification relaxes the required condition of~\cite{chen2023query} from ``$kR/P \le \varepsilon M$'' to ``$k/P \le \varepsilon M$''.
                
                \medskip
                \item {\bf Handling low-degree Steiner vertices.} Finally, we need to solve $(\sU_{low}, \sW_1)$. In their approach, this part takes $O(RP\cdot RPk)$ queries, and this is the step where our main improvement comes from.

                \medskip
                {\bf Our modification to this step.} By our improved bound for set cover (from Section~\ref{sec:set-cover}), the query complexity of this part reduces to $\wt{O}(k^2/M + RP (n + k))$. More precisely, to simulate the algorithm in \Cref{lem:rgmm}, we do the following.
                \begin{enumerate}
                    \item To run the \RGMM{} oracle, we need to sample $\wt{O}(1)$ elements from $\sU_i$ uniformly at random. Note that each element of $\sU_i$ corresponds to a small component (a component with a small number of representatives). To find a small component uniformly at random, we first pick a terminal uniformly at random and run a BFS to determine whether it lies in a small component, and if so, whether it is a representative terminal.
                    If the terminal is a representative and lies in a small component, we choose the corresponding connected component with probability $1/z$ where $z$ is the number of representative terminals in that connected component (note that BFS returns this number  as well). This approach ensures that each small connected component has an equal probability of being sampled. Moreover, due to the bound on the number of small connected components, i.e., $|\sU_i|\geq M$, we expect to encounter a terminal in a small connected component every $\wt{O}(k/M)$ samples. Therefore, the total cost of running all these BFS subroutines is $\wt{O}(k^2/M)$, since each BFS takes $\wt{O}(k)$ time.
                    
                    \item For a small component (a vertex in graph $H$ of \Cref{lem:rgmm}), we need to identify all the Steiner nodes within a distance of at most $\tau$ (which is set roughly as $(3/5) (1+\varepsilon)^i$). This step can be completed in $\wt{O}(n)$ time. A similar step with the same time complexity also appears in the proof of \Cref{lem:rgmm}.
                    
                    \item Note that the number of Steiner nodes within this close distance is at most $R$. Let $\hat{\mc{S}}$ be the set of these Steiner nodes. Next, the \RGMM{} algorithm requires a random neighbor of a small component. To get such a neighbor, we keep picking pairs $(v,t)$ in a random order, where $t\in T_{low}$ and $v\in \hat{\mc{S}}$, and querying their distance. The first time that we find a pair $(v,t)$ in close distance, we run a BFS from $t$ to check if it is a representative terminal and if it lies in a small component, which takes $\wt{O}(k)$ time. If it is not a representative terminal or does not lie in a small connected component, we skip this terminal. Otherwise, we return its small connected component with probability $1/z$, where $z$ is the number of representative terminals in that connected component, ensuring that all neighbors have an equal probability of being selected. We run the above procedure until we find a random neighbor. Therefore, using the running time from \Cref{lem:rgmm} (substituting $\alpha$ and $\beta$ for $R$ and $P$), and considering that we run the BFS at most $\wt{O}(RP)$ times (which corresponds to the maximum degree of $H$), we obtain an $\wt{O}(RP (n + k))$ time algorithm to estimate the size of the matching.

                \end{enumerate}
            \end{itemize}
        \end{itemize} 
    %\end{itemize}
\end{itemize}

\paragraph{Simulation of Step 3 and its query complexity.} 
As we are following the exact implementation of~\cite{chen2023query} for this step, the additional query complexity of this step (compared to the Step 2) is equal to $\wt{O}(nk/M)$ for both their algorithm and our algorithm.

\subsection{Query Complexity Analysis and Proof of Theorem \ref{thm:steiner-tree}}\label{sec:steiner-analysis}
For completeness, we start with the query complexity analysis of~\cite{chen2023query}.

\paragraph{Analysis of the query complexity of the algorithm of~\cite{chen2023query}.} As computed in~\Cref{sec:setp-2-query}, the overall query complexity of their algorithm is bounded by
\begin{align*}
 &\wt{O}(n\kappa + \frac{Mk}{\varepsilon} + \frac{nM}{\varepsilon} + \frac{k^2}{M} + \frac{nk}{R} + \frac{nk}{P} + \frac{k^2}{M} + (RP)^2 k + \frac{nk}{M}) \\
 = \;&\wt{O}(n\kappa + nM + \frac{nk}{R} + \frac{nk}{P} + (RP)^2 k + \frac{nk}{M}). &&\rhd \varepsilon =O(1), k\le n
\end{align*}

\noindent Furthermore, the conditions that need to be satisfied are 
\begin{itemize}
    \item $k\leq \kappa$, or 
    \item $k>M$ and $n/R<\wt{O}(\varepsilon M)$ and $k > \wt{\Omega}(P)$ and $kR/P\leq \varepsilon M$.  
\end{itemize}

\noindent The query complexity of their algorithm under the above conditions can be optimized by setting $\kappa = M = n^{6/7}$, $R = n^{1/7}$, and $P = n^{2/7}$, which gives the total query complexity of $\wt{O}(n^{13/7})$.

Now, we prove the main theorem of this section.
\begin{proof}[Proof of~\Cref{thm:steiner-tree}]
    As we are implementing the same algorithm as~\cite{chen2023query}, except replacing their set cover subroutine with a more efficient algorithm, the approximation analysis follows exactly from their proof. It only remains to bound the query complexity of our proposed algorithm for metric Steiner tree using the improved sublinear algorithm for set cover. In~\Cref{sec:setp-2-query}, we analyzed the query complexity of the component that is implemented differently in our algorithm. Now, we put the query complexity of all parts together and compute the overall complexity. 

    \paragraph{Analysis of the query complexity of our algorithm.} The overall query complexity of our algorithm is bounded by
    \begin{align*}
        &\wt{O}(n\kappa + \frac{Mk}{\varepsilon} + \frac{nM}{\varepsilon} + \frac{k^2}{M} + \frac{nk}{R} + \frac{nk}{P} + \frac{k^2}{M} + RP (k+n) + \frac{nk}{M}) \\
        =\;&\wt{O}(n\kappa + nM + \frac{nk}{R} + \frac{nk}{P} + RPn + \frac{nk}{M}) &&\rhd \varepsilon = O(1), k\le n
    \end{align*}

    Note again that the main difference is that the term $P^2R^2k$ is replaced by $RPn$. 
    Furthermore, the conditions that need to be satisfied are also slightly more relaxed, as follows:
    \begin{itemize}
        \item $k\leq \kappa$, or 
        \item $k>M$ and $n/R<\wt{O}(\varepsilon M)$ and $k > \wt{\Omega}(P)$ and $k/P\leq \varepsilon M$.
    \end{itemize}
    Specifically, ``$kR/P \le \varepsilon M$'' is replaced by ``$k/P \leq \varepsilon M$''.
    Then, our algorithm can be optimized by setting $\kappa = M = n^{2/3}, R = \wt{\Theta}(P) = \wt{\Theta}(n^{1/3})$ which gives the total query complexity of $\wt{O}(n^{5/3})$.
\end{proof}

\section*{Acknowledgment} The work was conducted in part while Sepideh Mahabadi and Ali Vakilian were long-term visitors at the Simons Institute for the Theory of Computing as part of the Sublinear Algorithms program.

\bibliographystyle{plainnat}
\bibliography{references}
	
\end{document}